\documentclass[11pt, letterpaper]{article}

\usepackage{chngcntr}
\usepackage{url}
\usepackage{dsfont}
\usepackage{bbm}
\usepackage[final]{graphicx}
\usepackage{amsmath}
\usepackage{varwidth}
\usepackage{amssymb,amsthm}
\usepackage{algorithm}
\usepackage{algorithmic}
\usepackage{appendix}
\usepackage{hyperref}
\usepackage{bm}
\usepackage{cite}
\usepackage{enumerate}
\usepackage[papersize={8.5in,11in},top=1in,left=1in,right=1in,bottom=1in]{geometry}
\usepackage{xfrac}
\usepackage{sidecap}

\makeatletter
\def\@seccntformat#1{\@ifundefined{#1@cntformat}%
   {\csname the#1\endcsname\quad}  
   {\csname #1@cntformat\endcsname}
}
\let\oldappendix\appendix 
\renewcommand\appendix{%
    \oldappendix
    \newcommand{\section@cntformat}{\appendixname~\thesection\quad}
}
\makeatother

\newenvironment{lessspaceitemize*}%
  {\begin{itemize}%
  \vspace{-1mm}
    \setlength{\itemsep}{0pt}%
    \setlength{\parskip}{0pt}}%
    {\vspace{-1mm} \end{itemize}}

\newenvironment{lessspaceenum*}%
  {\begin{enumerate}%
  \vspace{-1mm}
    \setlength{\itemsep}{0pt}%
    \setlength{\parskip}{0pt}}%
  {\end{enumerate}}

\counterwithout{figure}{section}

\newenvironment{definition}[1][Definition]{\begin{trivlist}
\item[\hskip \labelsep {\bfseries #1}]}{\end{trivlist}}

\DeclareMathOperator*{\argmax}{arg\,max}
\DeclareMathOperator*{\argmin}{arg\,min}

\newtheorem{theorem}{Theorem}[section]
\newtheorem{lemma}[theorem]{Lemma}
\newtheorem{claim}[theorem]{Claim}

\newtheorem{corollary}[theorem]{Corollary}
\newtheorem*{inf_thm*}{Informal Theorem}

\newenvironment{thm_app}[1]{\noindent\textbf{Theorem~\ref{#1}.}}{\par\addvspace{\baselineskip}}
\newenvironment{lem_app}[1]{\noindent\textbf{Lemma~\ref{#1}.}}{\par\addvspace{\baselineskip}}
\newenvironment{clm_app}[1]{\noindent\textbf{Claim~\ref{#1}.}}{\par\addvspace{\baselineskip}}

\newlength\myindent
\newlength\myindenttwo
\newcommand\bindent{%
  \begingroup
  \setlength{\itemindent}{\myindent}
  \addtolength{\algorithmicindent}{\myindent}
}
\newcommand\eindent{\endgroup}

\newcommand\bindenttwo{%
  \begingroup
  \setlength{\itemindent}{\myindenttwo}
  \addtolength{\algorithmicindent}{\myindenttwo}
}
\newcommand\eindenttwo{\endgroup}

\title{Posted Pricing sans Discrimination}
\author{Shreyas Sekar\\ Rensselaer Polytechnic Institute}

\begin{document}

\newpage

\maketitle

\begin{abstract}
	In the quest for market mechanisms that are easy to implement, yet close to optimal, few seem as viable as posted pricing. Despite the growing body of impressive results, the performance of most posted price mechanisms however, rely crucially on price discrimination when multiple copies of a good are available. For the more general case with non-linear production costs on each good, hardly anything is known for general multi-good markets. With this in mind, we study a Bayesian setting where the seller can produce any number of copies of a good but faces convex production costs for the same, and buyers arrive sequentially. Our main contribution is a framework for non-discriminatory pricing in the presence of production costs: the framework yields posted price mechanisms with $O(1)$-approximation factors for fractionally subadditive buyers, logarithmic approximations for subadditive buyers, and also extends to settings where the seller is oblivious to buyer valuations. Our work presents the first known results for Bayesian settings with production costs and is among the few posted price mechanisms that do not charge buyers differently for the same good.

\end{abstract}

	\section{Introduction}
In the quest for market mechanisms that are \emph{simple}, yet \emph{approximately optimal}~\cite{hartlineR09}, few candidates seem as appealing as \emph{posted price mechanisms}, where the seller simply posts prices on the items and buyers consume their desired bundles. This optimism towards posted prices is not without cause: a long line of research has established that in addition to a slew of favorable properties, mechanisms based on posted pricing provide near-optimal performance guarantees in a number of diverse applications~\cite{adamczykBFKL15,blumrosenH08,dobzinski07,feldmanGL15,meirCF13}. More recently, some of these results have been extended to general settings with combinatorial buyer valuations, where sellers only possess distributional information regarding the same~\cite{cohen2015invisible, feldmanGL15,hsu2015prices}; such results are a crucial first step towards the larger goal of understanding the performance of posted prices in more realistic, and general markets. 
	
	Real markets, however, often feature sellers who possess multiple copies of different goods or more generally, costs associated with producing or transporting each copy of good. Here, our picture of posted pricing mechanisms is far from complete. Papers dealing with multi-unit supply in combinatorial auctions have often resorted to \textit{price discrimination}, i.e., charging buyers differing prices for the same good, in order to extract good performance bounds~\cite{blumGMS11,chakrabortyEGMM10a,chawlaHMS10,cohen2015invisible,huang2015welfare}. In the posted pricing literature, one also encounters discriminative pricing policies under more benign labels such as \textit{dynamic pricing}~\cite{chakrabortyHK09}, \textit{sequential posted-pricing}~\cite{chawlaHMS10} or \emph{anonymous reserve prices}~\cite{daskalakisP11}. Looking beyond multi-unit supply to the more general case of production costs, little is known about the performance of posted price mechanisms when the underlying costs are non-linear, even less so in Bayesian settings.
	
	\begin{quote}
	\emph{Is discriminatory pricing essential for good performance in the presence of multi-unit supply? Do posted prices give desirable guarantees in a Bayesian setting involving production costs?} \end{quote}
	
	 These are the questions that we seek to answer in this work. Our work is motivated by both theory and practice; price discrimination, while accepted practice in some domains (airline tickets, taxis), is inherently unfair to the buyers and can lead to adverse effects in the long-run~\cite{anderson2010price}. On the other hand, extending the state-of-the-art in \emph{Bayesian Mechanism Design} to markets with general buyer valuation functions and production costs often necessitates blurring the line between good and copy, resulting in mechanisms that price each copy of a single good as if it were a distinct good. 

	Our central result answers both the questions posed above for Bayesian settings with fractionally subadditive (XoS) buyer valuations and convex production costs. Specifically, we design an incentive compatible, non-discriminatory mechanism based on static posted prices that extracts a constant fraction of the optimum social welfare. At a high level, our work presents a strong case for treating buyers fairly: by bounding the gains obtained via discriminatory pricing, we show that the incentives for discrimination may not offset the negative impact of treating buyers unfairly. More importantly, our results are enabled by a rather general black-box framework for non-discriminatory pricing in the presence of non-linear costs, which may find application in other markets featuring posted prices. We believe that this is the most useful contribution of our current work.\\

%
	
\noindent\textbf{Market Model} We consider a market where a single seller controls a set $\mathcal{I}$ of $M$ goods and $N$ buyers have combinatorial valuation functions: $v: 2^{\mathcal{I}} \rightarrow \mathbb{R}^+ \cup \left\lbrace 0 \right\rbrace$. The seller faces a convex production cost function with non-decreasing marginal costs for each good, i.e., for every $i \in \mathcal{I}$, the seller incurs a cost of $C_i(x)$ for producing $x$ units of this good. Such convex costs strictly generalize settings with limited supply. For the majority of this work, we assume a Bayesian setting with prior $\mathcal{F}=F_1 \times F_2 \ldots \times F_N$: buyer $i$'s private valuation is drawn independently from distribution $F_i$, which is known to the seller. Our objective is to maximize the social welfare, which is defined as the total value derived by the buyers minus the production cost incurred by the seller. Towards this goal, we propose two types of posted pricing mechanisms depending on whether or not the seller has to commit to producing an exact amount of each good in advance. Both these mechanisms follow a simple template: $(i)$ the seller posts a single price per good, and $(ii)$ buyers arrive sequentially and each buyer purchases a utility maximizing subset of the available goods. 

\subsubsection*{Challenges: Keeping Prices Static in the face of Non-Linear Production}
Given the voluminous body of research in mechanism design and algorithmic pricing, the absence of mechanisms that encompass markets with production (specifically with a view towards fair pricing) is somewhat conspicuous. This is particularly surprising given that increasing marginal costs are a natural consideration in a number of markets. As~\cite{blumGMS11} points out,

\begin{quote}
``\textit{the unlimited-supply case is too optimistic and the limited-supply case too pessimistic.  More often than not, additional sources can be found, but at higher cost.}''
\end{quote}

Production costs however, introduce fundamentally new challenges to the welfare maximization problem with complex buyer valuations~\cite{blumGMS11}. For example, tailoring prices to carefully match the buyer valuations (as in existing work~\cite{feldmanGL15}) may be impossible without price discrimination when buyers have vastly differing valuations for the same good. On the contrary, just mimicking the production cost (as in~\cite{blumGMS11,huang2015welfare}) may result in the goods being sold to an initial sequence of low-valuation buyers. In order to overcome these obstacles, we develop a new framework that is a product of examining the \textit{profit-surplus equivalence} idea used in~\cite{dobzinski07,feldmanGL15} at a granular level and re-building it from the ground up in greater generality. In the process, we generalize the central result of~\cite{feldmanGL15} to settings with production, and also show that one can strictly improve upon the results in~\cite{blumGMS11, huang2015welfare} in Bayesian settings.

\subsection{Summary of Contributions}
As mentioned previously, we introduce two types of static posted pricing mechanisms in this work. In both of these, the seller fixes a single price per good, and buyers arrive sequentially in some arbitrary, possibly adversarial order and each buyer purchases a subset of the available goods that maximizes her utility at the originally fixed prices. 

\begin{enumerate}
\item \textbf{On the Fly Mechanisms} (OTF): The seller fixes an upper bound on the number of copies of each good that he is willing to produce and incurs a production cost only for the items that are actually sold. Once the number of copies of a good sold reaches the upper bound, the item becomes unavailable for future buyers. 

\item \textbf{Commitment Mechanisms:} The seller commits to producing a certain quantity of each good and incurs the production cost for that quantity whether or not those items are sold. 
\end{enumerate}

We observe that when the seller has limited supply, the two mechanisms are equivalent. However, for strictly non-linear production  costs, on the fly mechanisms are a relaxation of commitment mechanisms, and therefore, the latter is harder to design. Finally, all of the mechanisms in this work are trivially \emph{dominant strategy incentive compatible}.

\subsubsection*{Central Result: $O(1)$-approximate OTF Mechanism for Bayesian Settings}
We begin by stating our black-box reduction that transforms any approximation algorithm $Alg$ for the \emph{allocation problem}:  the problem of allocating goods to XoS buyers to maximize social welfare, into a non-discriminatory posted price mechanism that ensures a social welfare of $\frac{1}{2}E_{\vec{v} \sim \mathcal{F}}[SW(Alg(\vec{v}))]$, where $SW(Alg(\vec{v}))$ is the social welfare guaranteed by $Alg$ on input valuations $\vec{v}$.

\begin{inf_thm*}
\textbf{(On the fly mechanism)} Given any $c$-approximation algorithm ($c \geq 1$) for settings with XoS valuations and convex production costs, we can compute a posted price mechanism that ensures a social welfare that is within a factor $2c$ of that of the optimum allocation. 
\end{inf_thm*}

The black-box result makes use of both a demand oracle and an XoS oracle for the buyer valuations as well as access to the distributions from which these valuations are derived. To supplement the reduction, we also present a simple, poly-time two approximation algorithm for the allocation problem with $XoS$ buyers and convex costs. Together, they yield our main computational result: ``\emph{a mechanism that posts a single price per good and extracts one-fourth of the optimum welfare for arbitrary buyer arrival orders}". 


\noindent \textbf{(Why is this result significant?)} This result directly extends the posted price mechanisms in~\cite{feldmanGL15} which were used to obtain good welfare bounds for unit supply settings. Moreover, while previous mechanisms~\cite{blumGMS11, huang2015welfare} gave non-trivial guarantees only for specific cost functions, our results hold for all convex production cost functions. Although the results in~\cite{feldmanGL15} directly apply to settings with increased supply, this would require the seller to post a different price for multiple copies of the same item. On the contrary, the notion of posting a \emph{single price for all copies of a good} is the driving force behind all of our mechanisms. Finally, our $2$-approximation algorithm for the allocation problem is also the only known algorithm for settings with production. 

%

\noindent \textbf{Commitment Mechanism and XoS valuations:}
For the somewhat challenging class of commitment mechanisms, we provide a slightly modified mechanism whose welfare guarantees are sensitive to the \emph{structure of the allocation} returned by the algorithm: our mechanism yields desirable welfare guarantees for instances where the algorithm returns a solution whose expected social welfare is large compared to the cost incurred. Formally, we define an instance to be $\alpha$-bounded with respect to an algorithm $Alg$, if in the solution returned by the algorithm, the expected value derived by the buyers is at least a factor $\alpha$ times the expected production cost incurred by the seller. 

\begin{inf_thm*}
\textbf{(Commitment Mechanisms)} Given an algorithm $Alg$ for the allocation problem, and an instance that is $\alpha$-bounded with respect to $Alg$ for $\alpha \in [2,\infty]$, there exists a posted price mechanism that ensures a social welfare of $\frac{1}{2}\frac{\alpha- 2}{\alpha-1}E_{\vec{v} \sim \mathcal{F}}[SW(Alg(\vec{v}))]$. 
\end{inf_thm*}	

The algorithm requires access to the same oracles as before. Similar notions of $\alpha$-boundedness have been considered in the literature for other problems involving costs~\cite{feigeIMN13}. We argue that $\alpha$ is a natural parameter that arises in the context of commitment mechanisms as the seller may not be willing to invest in markets where he does not recover a large enough social value. In Section~\ref{sec:commitment}, we also show that if the cost functions are `sufficiently convex', then $\alpha$ is large irrespective of the buyer valuations, and our mechanism provides good bounds. Finally, in conjunction with our $2$-approximation algorithm, we get a  $4 \frac{\alpha - 1}{\alpha-2}$-approximate commitment mechanism for XoS buyers. 

\subsubsection*{Extensions: OTF Mechanisms for Unknown Distributions and Subadditive Buyers}
We also partially generalize our results to settings with more general buyer valuations. Specifically, for markets with subadditive buyer valuations, a strict generalization of XoS functions, and limited supply, \emph{we design posted pricing mechanisms that obtain a $O(\log M)$-approximation to the optimum social welfare}. Once again, the selling point of our mechanism is that it only posts a single price per good. 

Given the fact that subadditive valuations are log-approximate XoS functions, an $O(\log(M))$-approximate mechanism appears to be a natural by-product of our previous result. Indeed, in the conclusion of~\cite{feldmanGL15}, it was mentioned that their techniques imply the existence of a log-approximate posted pricing mechanism for settings with subadditive buyers and unit-supply.  Our result shows that it is possible to efficiently compute these posted prices in rather general settings with limited supply.

\noindent\textbf{Unknown Buyer Valuations and Production Costs:} All of the above results remain valid as long as we have query access to $(i)$ buyer distributions, $(ii)$ demand oracles and $(iii)$ XoS oracles. These are standard assumptions in the mechanism design literature. It is however, natural to envisage settings where the seller may not possess any prior information regarding the buyer valuations; in that case, the other oracles do not serve any purpose. Against this backdrop, we pose the following ambitious question: \emph{Is it possible to design static posted pricing mechanisms in markets with production costs when the seller has absolutely no information about buyer valuations?}

Surprisingly, we answer this question in the affirmative and show that it is possible to extract good social welfare even when the seller is \emph{completely in the dark}. Concretely, we present an on the fly mechanism for markets with XoS buyers and production costs that achieves a $O(\log M \log N)$-approximation to the optimum social welfare. Previously, such results were only known for dynamic pricing mechanisms~\cite{blumGMS11,chakrabortyHK09,huang2015welfare}. 

\noindent\textbf{Remark: Query Access to Buyer Valuations:}
Our mechanisms for the Bayesian setting exhibit a runtime polynomial in the support size of the product distribution $\mathcal{F}$. However, one can convert these mechanisms into computationally efficient procedures in a straightforward manner by repeatedly sampling valuation profiles and computing allocations for the sampled valuations. Specifically, directly applying the sampling techniques in~\cite{feldmanGL15}, we obtain mechanisms that run in time polynomial in $N,M,$ and $\frac{1}{\epsilon}$ with an additive welfare degradation of $\epsilon$ from the stated bound. In order to clearly illustrate the technical ideas underlying our theorems and avoid messy notation, we explicitly suppress the sampling aspect of our mechanisms for the rest of this work.

%
%
%
%


\subsection{Related Work}

This paper bridges the gap between two closely related fields: auction theory, and envy-free algorithmic pricing. Although the mechanisms in this work technically fall under the purview of Bayesian mechanism design~\cite{chawlaS14}, they are not auctions in the traditional sense as they do not elicit any input from the buyers. In fact, all of our mechanisms are trivially dominant strategy incentive compatible (DSIC) as the buyer does not interact with the mechanism in any way other than to purchase a utility-maximizing bundle at the fixed prices. In that sense, our results fall under the realm of \emph{item pricing algorithms}. Traditionally, papers in this area have looked at relaxations of Walrasian equilibrium in full information settings where buyers purchase goods simultaneously instead of sequentially~\cite{chenR08,fuKL12,gul1999walrasian}. However, prices that extract high welfare exist only for simple valutions~\cite{feldmanwelfarepricing}. In that context, one could view the pricing mechanisms in this work as Walrasian equilibria obtained by relaxing the \emph{simultaneous arrival} constraint.

To the best of our knowledge, there are no mechanisms in the literature that meet all of the desiderata fulfilled by our framework: $(i)$ general buyer valuations (XoS), $(ii)$ (convex) production cost functions, $(iii)$ non-discriminatory pricing, $(iv)$ non-identically distributed buyers, and $(v)$ dominant strategy incentive compatibility. However, if we relax some of these constraints, we stumble upon a number of closely related works. For instance, for single-parameter or even unit-demand settings, there exist computationally efficient mechanisms that maximize social welfare~\cite{nisan2007algorithmic}, whereas anonymous reserve prices are known to yield good performance in i.i.d environments. In~\cite{feldmanGL15}, the authors design a $O(1)$-posted pricing mechanism for XoS valuations and unit-supply; we successfully generalize their results without having to resort to price discrimination. 

Owing to their popularity in recent applications, there has been a surge in the theoretical investigation of dynamic or online pricing mechanisms. Of particular relevance to us are~\cite{azar16convex,blumGMS11,huang2015welfare}, all of which look at dynamic pricing via a OTF-like mechanism in the face of production. Such mechanisms are more powerful than ours, however, they do not guarantee fairness as different buyers may observe different prices for the same good depending on available supply. Moreover, the approximation guarantees offered in these papers depend on the nature of the production cost function; we resort to the Bayesian setup to ensure that the welfare bounds of our OTF mechanisms are independent of the function. Conceptually, our model is closely linked to the burgeoning body of partition mechanisms (as defined in~\cite{rubinstein16}), a strict generalization of the posted pricing mechanisms in this work where sellers post prices on bundles. Such mechanisms are known to perform favorably in a number of settings with respect to both welfare~\cite{feldmanGL16,parkesU00} and revenue~\cite{rubinstein16}. In comparison, our work shows that it is often possible to achieve good welfare using only item prices. 

A recurring theme in mechanism design concerns the use of black-box reductions to transform approximation algorithms into pricing mechanisms, as we do in this work. Although this has yielded a number of beautiful computational tools, such mechanisms are usually discriminatory and relax DSIC in favor of weaker notions such as Bayesian incentive compatibility~\cite{beiH11}. Finally, in contrast to our focus on convex production cost functions, in~\cite{BlumMY15}, the authors study economies of scale (concave costs) and provide non-discriminatory posted pricing mechanisms for the Bayesian regime . 

\section{Model and Preliminaries}

We consider settings where a seller controls a set $\mathcal{I}$ of goods; each $i \in \mathcal{I}$ has a convex production cost, i.e., for any $n$, the seller incurs a (marginal) cost of $c_i(n)$ for producing the $n^{th}$ copy of this good. Moreover, since the costs are convex, $c_i(n)$ is non-decreasing as $n$ increases. The seller can choose to produce any number of goods for which he faces the corresponding aggregate production cost ($C_i(n) = \sum_{r=1}^n c_i(r)$ for $n$ copies of good $i$). Notice that convex production costs are a strict generalization of limited or multi-unit supply. For instance $c_i(n) = 0$ for $n \leq k$ and $c_i(n) = \infty$ otherwise, indicates a fixed supply of $k$ units. 

Every buyer $j$ has a valuation function with $v_j(S)$ being her value for a set $S \subseteq \mathcal{I}$ of goods. Each valuation $v_j$ is monotone with $v_j(\emptyset) = 0$. Although, this only allows for the buyer to  purchase one copy of each good, all of the results in this work extend directly to the case when buyers' valuations allow for multisets. Given an allocation $\vec{S} = (S_1, S_2, \ldots, S_N)$ where $S_j$ is the set of goods allocated to buyer $j$, suppose that $\forall i \in \mathcal{I}$, $x_i$ represents the number of units of this good allocated to the buyers. Then, the total social welfare of this allocation is $SW(\vec{S}) = \sum_{j=1}^N v_j(S_j) - \sum_{i \in \mathcal{I}}C_i(x_i)$.

We are interested in the social welfare maximization problem in a Bayesian setting, where the buyer valuations are drawn from the distribution $\mathcal{F} = F_1 \times F_2 \ldots \times F_N$. We refer to the case where each $F_i$ is deterministic as the full information version. For the majority of this work, we will assume that all buyers have XoS valuations, which strictly generalize submodular valuations and are a special class of subadditive valuations. 

\begin{description}
\item [Fractionally Subadditive (also called `XoS') Valuations] $\exists$ a set of additive functions\\$(a_1, \ldots, a_r)$ such that for any $T \subseteq \mathcal{N}$, $v(T) = \max_{j=1}^{r} a_j(T)$. These additive functions are referred to as {\em clauses.} Recall that an additive function $a_j$ has a single value $a_j(i)$ for each $i \in \mathcal{N}$ so that for a set $T$ of agents, $a_j(T) = \sum_{i \in T}a_j(i)$. 

\item [Subadditive Functions] A valuation function $v$ is said to be subadditive if for every $S,T \subseteq \mathcal{I}$, $v(S \cup T) \leq v(S) + v(T)$.

\end{description}

%

\noindent\textbf{Oracle Access} The standard approach in the literature while dealing with set functions (where the input representation is exponential in size) is to assume the presence of an oracle that allows indirect query access to the valuation. In this work, unless mentioned otherwise, we assume that the valuations are accessed via $(i)$ a \emph{demand oracle} that when queried with a vector of payments $\vec{p}$ returns a set $S$ that maximizes the quantity $v(S) - \sum_{i \in S}p_i$, and $(ii)$ an \emph{XoS oracle} that for an XoS function $v$ and a set $T \subseteq \mathcal{I}$ returns the additive clause $a_l$ that maximizes $a_l(T)$. 

\subsubsection*{Seller's Mechanism: Static Posted Prices}
A posted pricing mechanism with item prices is said to \emph{static} or \emph{anonymous} or \emph{non-discriminatory} if no two buyers are offered the same good at differing prices. We now formally define the two types of mechanisms considered in this work.

\textbf{On the fly mechanism}
In an `on the fly' mechanism, a seller only produces as much of the good as required by the buyers.

\begin{enumerate}

\item the seller posts a single price per good, and fixes an upper bound on the number of copies of each good that he is willing to produce.

\item buyers arrive in some arbitrary order

\item each buyer looks at the set of available goods and purchases her utility maximizing bundle.

\item an item becomes unavailable once the number of copies purchased equals the upper bound.
\end{enumerate}
The seller incurs the cost for producing the good only when a buyer decides to purchase a good. The upper bound can be thought of as the seller preparing his machinery and other paraphernelia for producing a certain number of goods, but stopping short of actually producing them.

\textbf{Commitment Mechanisms} On the contrary, in a commitment mechanism, the seller produces the goods before the buyer arrives and has to make a decision on the quantity of each good that he has to manufacture. Therefore, the seller incurs an additional cost for any goods that are not sold. From a design perspective, these mechanisms are much harder than on the fly mechanisms.

\begin{enumerate}

\item the seller posts a single price per good and produces a certain number of copies of each good.

\item buyers arrive in some arbitrary order

\item each buyer purchases her utility maximizing bundle from the remaining items.
\end{enumerate}

\section{On the fly mechanism for XoS buyers}
\label{sec:framework}
In this section, we present our main theoretical result: a black-box reduction that allows us to transform any approximation algorithm for the allocation problem into a posted price mechanism with similar welfare guarantees. Along the way, we design a general framework for posted pricing in the presence of production and present it as a sequence of simple ideas. The black-box reduction functions as an existence result. Following this, we present an actual $2$-approximation algorithm for the problem of allocating goods to XoS buyers when faced with production costs. Together, these yield a posted pricing mechanism that provides a $4$-approximation to the optimum solution.

\begin{theorem}
	\label{thm:otfbayesin}
	Given a product distribution $\mathcal{F}$, an algorithm $Alg$, demand and XoS oracles, we can compute in poly-time a price vector $\vec{p^*}$ so that the on the fly mechanism with this price vector provides a welfare guarantee (in expectation) of $\frac{1}{2}E_{\vec{v} \sim \mathcal{F}} [SW(Alg(\vec{v}))]$.
\end{theorem}

We reiterate an important point made earlier regarding the computational nature of our mechanisms. In order to not detract from main ideas involved in the proof, we assume that the distribution $\mathcal{F}$ can be represented using a polynomial number of valuation profiles. When this is not the case, one can use standard sampling techniques as in~\cite{feldmanGL15} to compute prices in time polynomial in $\frac{1}{\epsilon}$ so that the mechanism provides a welfare guarantee of $\frac{1}{2}E_{\vec{v} \sim \mathcal{F}} [SW(Alg(\vec{v}))] - \epsilon$. Before proving the theorem, we state the concomitant results that lead to an actual computational mechanism.

\subsubsection*{Combinatorial Algorithm for Allocating Goods to XoS buyers}
The secondary algorithmic contribution of this paper is a $2$-approximation algorithm for the `welfare maximizing allocation' problem that complements (and in some sense, enables) the black-box transformation in Theorem~\ref{thm:otfbayesin}. Our algorithm generalizes the deterministic $2$-approximation algorithm for XoS buyers and unit-supply first studied in~\cite{dobzinskiNS10} to settings with production costs. Formally, the allocation problem consists of a set $\mathcal{I}$ of goods and a profile $\vec{v} = (v_1, v_2, \ldots, v_N)$ of XoS functions along with production costs for each good; the objective is to allocate a set of goods to each buyer in order to maximize the resulting social welfare. 

We now state the claim corresponding to our algorithm that in conjunction with Theorem~\ref{thm:otfbayesin} yields a $4$-approximate posted price mechanism. 

\begin{claim}
\label{clm_allocationalg}
 ($2$-Approximation Algorithm) Given a valuation profile  of XoS functions, and access to demand, and XoS oracles, we can compute in poly-time an allocation of goods to the buyers whose social welfare is at least half of that of the welfare maximizing allocation.
\end{claim}

\begin{algorithm}[htbp]
\caption{Combinatorial Algorithm for allocating goods to XoS buyers}
\label{alg_genprocedurebody}
\renewcommand{\algorithmicrequire}{}
\renewcommand{\algorithmiccomment}[1]{// \textit{#1}}
\begin{algorithmic}[1]
\vspace{2mm}
\REQUIRE \textbf{Initialization Phase}
\vspace{2mm}
\bindent
\STATE $S_0 \gets \mathcal{I}$; $S_1, S_2, \ldots, S_N \gets \emptyset$ \COMMENT{Initial Allocation} 
\STATE Set $p_i = c_i(1)$ for all $i \in \mathcal{I}$,  \COMMENT{Initial Prices}
\STATE Set $X_i = \{0\}$ for all $i \in \mathcal{I}$ \COMMENT{$X_i$: set of buyers to whom good $i$ has been allocated.}
\STATE Set $q_j(i) = 0$ for all goods $i$, buyers $j$ \COMMENT{$q_j(i)$: price assigned to item $i$ by buyer $j$}
\STATE Set $y_i = 0$ for all $i \in \mathcal{I}$ \COMMENT{$y_i$: identity of buyer offering good $i$ at the cheapest price}
\STATE Set initial prices on the items, $p_t = p^*_t$.
\eindent
\vspace{2mm}
\REQUIRE \textbf{Allocation Phase}: For each buyer $i \in 1,2, \ldots, N$
\vspace{2mm}
\bindent
\STATE Let $S'_i$ be the bundle demanded by buyer $i$ at prices $\vec{p}$; set $S_i = S'_i$
\STATE Let $(x^i_j)_{j \in S'_i}$ denote the maximizing XoS clause for buyer $i$ and set $S'_i$.
\STATE \textbf{Updation:} For all $j \in S'_i$:  
\bindenttwo
\STATE If $y_j \neq 0$, $S_{y_j} = S_{y_j} \setminus \{j\}$, $X_j = X_j \cup \{i\} \setminus \{y_j\}$.
\STATE $q_i(j) = x^i_j$ and $q_0(j) = c_j(|X_j|)$. 
\STATE Set $p_j = \min_{i \in X_j} q_i(j)$; set $x_j = \argmin_{i \in X_j} q_i(j)$.
\eindenttwo
\eindent
\end{algorithmic}
\end{algorithm}

The obvious similarities between the algorithm and a (dynamic) posted-price mechanism are hard to miss. The functioning of Algorithm~\ref{alg_genprocedurebody} can interpreted as follows: the market consists of $N+1$ buyers (buyer $0$ is the seller) who arrive sequentially. Each new buyer can acquire their maximum utility bundle by purchasing each item from either the seller (at the current marginal production) or from other buyers (at a price derived from the XoS clause corresponding to their original bundle). We defer the proof of the approximation guarantee to the Appendix.

The main implication of Theorem~\ref{thm:otfbayesin} and Claim~\ref{clm_allocationalg} is a computationally efficient posted pricing mechanism whose welfare guarantee is one-fourth that of the optimum allocation.

\begin{corollary}
We can compute in poly-time a price vector $\vec{p^*}$ and supply vector $\vec{k^*}$ such that for any arrival order, the expected welfare of the posted price mechanism with these prices is at least one-fourth of the optimum social welfare.
\end{corollary}

\noindent\textbf{Proof of Theorem~\ref{thm:otfbayesin}} Before delving into the actual proof, we derive our framework for posted pricing in the presence of production. The framework consists of two ingredients: first, we provide a black-box technique to derive prices for a single good that lead to high welfare in the presence of production. Second, we provide some sufficient conditions on the posted prices that result in good mechanisms even for general combinatorial valuations. In later sections, we show how to bridge these two components to get a good mechanism as well as extend it to the Bayesian setting.

\subsubsection*{Framework - Part 1: How to price a single item when there are $k$ copies}
\label{framework}
Consider a full information market with a single good $i$. Suppose we allocate $k$ units of this goods to buyers whose valuations for the good are $v_1 \geq v_2 \geq \ldots \geq v_k$. The total welfare due to this particular allocation is $\sum_{n=1}^k [v_n - c_i(n)]$. We are not interested in negative welfare, so we assume that $\sum_{n=1}^k v_n \geq \sum_{n=1}^k c_i(n).$ Define the following functions: 
$$f_i(p) = \sum_{n=1}^k [v_n - p], \quad \text{and} \quad \pi_i(p) = pk - \sum_{n=1}^k c_i(n).$$

$f_i(p)$ is the buyer surplus that can be obtained at a price of $p$ if the same $k$ buyers as above decide to purchase the good, and $\pi_i(p)$, the profit (revenue minus cost) if all $k$ copies are sold (to any buyers). Next, we observe that both the functions are continuous, monotone, and $f_i(0) \geq \pi_i(0)$ whereas $f_i(v_1) \leq \pi_i(v_1)$. So, there must exist some price $p^*_i \in [0,v_1]$ satisfying the following condition.

\begin{equation}
\label{eqn_profsurp}
\text{(Profit-Surplus Equivalence)} f_i(p^*_i) = \pi_i(p^*_i).
\end{equation}

We now mention the simple yet crucial lemma that characterizes this price $p^*_i$.

\begin{lemma}
\label{lem_revsocwel}
The total profit if all $k$ items are sold at price $p^*_i$ is exactly half the social welfare due to $i$, i.e., $\pi_i(p^*_i)= kp^*_i - \sum_{n=1}^k c_i(n) = \frac{1}{2}\left(\sum_{n=1}^k [v_n - c_i(n)]\right)$.
\end{lemma}
\begin{proof}
The social welfare due to $i$ (taking the allocation for granted) is $SW(i) = \sum_{n=1}^k [v_n - c_i(n)]$. Therefore, 
$$ SW(i) = \sum_{n=1}^k [v_n - p^*_i] + \sum_{n=1}^k [p^*_i - c_i(n)] = f_i(p^*_i) + kp^*_i - \sum_{n=1}^k c_i(n) = \pi_i(p^*_i) + \pi_i(p^*_i).$$ 

The final equation comes from the Profit-Surplus Equivalence at $p^*_i$.\end{proof}

The next lemma states that even if we don't sell all $k$ items, the seller's profit at this price cannot be negative. We defer its proof to the Appendix.

\begin{lemma}
\label{lem_posprofit}
For all $t \leq k$, the profit that the seller makes by selling $t$ copies of good $i$ at price $p^*_i$ is at least the cost of producing these $t$ copies, i.e., $t p^*_i \geq \sum_{n=1}^t c_i(n)$.
\end{lemma}

\subsubsection*{Framework - Part 2: Sufficient Conditions for Good Mechanisms}
Previously, we derived a simple procedure for computing prices that satisfy certain `nice' conditions. Setting aside the problem of price computation, we now try to understand how one can utilize these prices to design mechanisms with high social welfare for XoS buyer valuations. Suppose that $Alg$ is some given algorithm for the allocation problem and that $Alg(\vec{v}) = (A_1(\vec{v}), A_2(\vec{v}), \ldots, A_N(\vec{v}))$ denotes the output allocation for valuation profile $\vec{v}$. Moreover, suppose that under this allocation, $N_i(\vec{v})$ represents the set of buyers to whom good $i$ is allocated, and $k_i(\vec{v}) = |N_i(\vec{v})|$. Fixing $\vec{v}$, let $o_j$ be the (maximizing) XoS clause for buyer $j$ in the allocation $Alg(\vec{v})$, i.e., $o_j$ is the additive clause that maximizes $o_j(A_j(\vec{v}))$.  Then, we can define the total value derived from good $i$, $V_i(\vec{v}) = \sum_{j \in N_i(\vec{v})}o_j(i)$, and also $V_i = E_{\vec{v} \sim\mathcal{F}}[V_i(\vec{v})]$. Similarly, we can also define the quantity representative of social welfare $SW_i(\vec{v}) = \sum_{j \in N_i(\vec{v})}o_j(i) - \sum_{n=1}^{k_i(\vec{v})}c_i(n)$, and its expected value $SW_i$. The following is our main structural lemma that holds for XoS combinatorial valuations derived from production distribution $\mathcal{F}$.

\begin{lemma}
\label{lem_structural}
Given an approximation algorithm $Alg$, suppose that $\exists$ a price vector $\vec{p^*}$ and a supply vector $\vec{k^*}$, satisfying the following conditions
\begin{enumerate}
\item for every good $i$, $p^*_ik^*_i - C_i(k^*_i) \geq 0$.

\item for every good $i$, $V_i - p^*_iE_{\vec{v} \sim \mathcal{F}}[k_i(\vec{v})] \geq \frac{1}{\alpha}SW(i)$

\item for every good $i$, $p^*_ik^*_i - C_i(k^*_i) \geq \frac{1}{\alpha}SW(i)$

\end{enumerate}
Then, an on the fly mechanism with prices $(p^*_1, p^*_2, \ldots, p^*_M)$ and supply $(k^*_1, \ldots, k^*_M)$ on the goods results in a welfare of $\frac{E_{\vec{v} \sim \mathcal{F}}[SW(Alg(\vec{v}))]}{\alpha}$.  

\end{lemma}
\begin{proof}
Fix some valuation profile $\vec{v} \sim \mathcal{F}$, and a buyer $j$. Suppose that for this valuation, $Sold_j(\vec{v})$ is the set of goods that are sold out (supply limit reached) when buyer $j$ arrives. Finally, consider another valuation profile $\vec{a}_{-j}$ for buyers other than $j$, so that this valuation is independent of $\vec{v}$. Define $\vec{a} = (v_j, a_{-j})$, and let $A_j(\vec{a})$ be the set of goods allocated to $j$ by $Alg$ for $\vec{a}$ with the corresponding (maximizing) XoS clause being $o_j$, i.e., $v_j(A_j(\vec{a})) = o_j(A_j(\vec{a}))$.

Under the fixed valuation $\vec{v}$ and for each $\vec{a}_{-j}$, the buyer $j$ could have purchased the set of goods in $A_j(\vec{a}) \setminus Sold_j(\vec{v})$. Therefore, the utility of this buyer under the fixed valuation is at least 

$$E_{\vec{a}_{-j}}[\sum_{i \in A_j(\vec{a}) \setminus Sold_j(\vec{v})} o_j(i) - p^*_i] = E_{\vec{a}_{-j}}[\sum_{i \in \mathcal{I}}\mathbbm{1}{(j \in N_i(\vec{a}))}.  \mathbbm{1}(i \notin Sold_j(\vec{v})).o_j(i) - p^*_i].$$ 

Adding up this quantity for all buyers and taking the expectation under $\mathcal{F}$, we get that the expected total surplus (buyer utility) is

\allowdisplaybreaks
\begin{align*}
E_{\vec{v} \sim \mathcal{F}}[\sum_{j =1}^N u_j(\vec{v})] & \geq \sum_{i \in \mathcal{I}}\sum_{j =1}^N E_{\vec{v}, \vec{a}_{-j}}[\mathbbm{1}{(j \in N_i(\vec{a}))}.  \mathbbm{1}(i \notin Sold_j(\vec{v})).o_j(i) - p^*_i] \\
& = \sum_{i \in \mathcal{I}}\sum_{j =1}^N E_{\vec{v}}[\mathbbm{1}(i \notin Sold_j(\vec{v}))] E_{\vec{a}} [ \mathbbm{1}(j \in N_i(\vec{a})).o_j(i) - p^*_i ] \\
& \geq \sum_{i \in \mathcal{I}}Pr(\text{i is not sold out}) E_{\vec{v}}[V_i(\vec{v}) - k_i(\vec{v})p^*_i] \\ 
& = \sum_{i \in \mathcal{I}}Pr(\text{i is not sold out}) (E_{\vec{v}}[V_i(\vec{v})] - p^*_iE_{\vec{v}}[k_i(\vec{v})]) \\ 
& \geq \frac{1}{\alpha}\sum_{i \in \mathcal{I}}Pr(\text{i is not sold out})SW(i).
\end{align*}

The second line comes from the observation that $Sold_j(\vec{v})$ does not depend on $v_j$ or $\vec{a}_{-j}$ whereas the other terms such as $N_i(\vec{a})$ and $o_j(i)$ depend only on $\vec{a}$. The final equality follows from condition (2) of the lemma statement. 

Moving on to the profit, for any fixed valuation $\vec{v}$, let $S(\vec{v})$ and $U(\vec{v})$ be the set of items that are sold out (all $k^*_i$ copies), and the ones that are not respectively. Suppose that for any $\vec{v}$, $t_i(\vec{v})$ is the number of copies of good $i$ sold by the mechanism. Then, the expected profit is at least $$E_v [\sum_{i \in S(\vec{v})}(k^*_ip^*_i - \sum_{n=1}^{k^*_i}c_i(n)) + \sum_{i \in U(\vec{v})}t_i(\vec{v})p^*_i - \sum_{n=1}^{t_i(\vec{v})}c_i(n))] .$$

We claim that for every $\vec{v}$, the profit due to the unsold goods is non-negative. This is a consequence of the first condition of the lemma statement according to which for every good $i$, $p^*_i \geq \frac{C_i(k^*_i)}{k^*_i}$. Moreover, as per the workings of on the fly mechanisms, $t_i(\vec{v})$ is always smaller than or equal to $k^*_i$. Finally, using the convexity of the cost function, we have that for all $1 \leq t \leq k^*_i$, $p^*_i \geq \frac{C_i(k^*_i)}{k^*_i} \geq \frac{C_i(t)}{t}$.

Therefore, we can safely ignore the profit due to the unsold items. The expected profit is then bounded from below by $\sum_{i \in \mathcal{I}}E_v[\mathbbm{1}(\text{i is sold out})k^*_ip^*_i - \sum_{n=1}^{k^*_i}c_i(n)] = \sum_{i \in \mathcal{I}}Pr(\text{i is sold out})(k^*_ip^*_i - C_i(k^*_i))$. Making use of the third condition in the lemma statement, we get that 
$$\sum_{i \in \mathcal{I}}Pr(\text{i is sold out})(k^*_ip^*_i - C_i(k^*_i)) \geq \frac{1}{\alpha}\sum_{i \in \mathcal{I}}Pr(\text{i is sold out})SW(i).$$

Now, we can add up our lower bounds on the expected surplus and expected profit to obtain an aggregate lower bound on the expected social welfare of
$$\frac{1}{\alpha}\sum_{i \in \mathcal{I}}Pr(\text{i is not sold out})SW(i) + \frac{1}{\alpha}\sum_{i \in \mathcal{I}}Pr(\text{i is sold out})SW(i) = \frac{1}{\alpha}E_{\vec{v} \sim \mathcal{F}}[SW(Alg(\vec{v}))].$$

Using the final condition in the lemma statement allows us to complete the proof. \end{proof}

\subsubsection*{Final leg: Showing the Black-Box Result}
Armed with our framework from the previous sections, we are now ready to prove our black-box result that takes as input a distribution $\mathcal{F}$ and an allocation algorithm $Alg$, and shows how to compute prices that result in an OTF mechanism with desirable guarantees. Specifically, we will compute prices based on the first part of the framework and then show that a mechanism with those prices satisfies the conditions of Lemma~\ref{lem_structural}. We use the same notation as defined previously.


Fix a valuation profile $\vec{v}$: for any good $i$, we can compute the price $p^*_i(\vec{v})$ that results in profit-surplus equivalence with respect to $Alg(\vec{v})$. Explicitly solving for Equation~\ref{eqn_profsurp}, we get that $p^*_i (\vec{v}) = \frac{1}{2k_i(\vec{v}))} \{V_i(\vec{v}) + C_i(k_i(\vec{v}))\}$.  Taking the expectation over $\mathcal{F}$ yields the actual parameters for our mechanism.

\begin{align*}
k^*_i & = E_{\vec{v} \sim \mathcal{F}}[k_i(\vec{v})] \quad  \text{and} \quad p^*_i = \frac{V_i}{2k^*_i}+ \frac{E_{\vec{v} \sim \mathcal{F}}[C_i(k_i(\vec{v})]}{2k^*_i} & i \in \mathcal{I}
\end{align*}

We also make the simplifying assumption that for every good $i$, $V_i - E_{\vec{v} \sim \mathcal{F}}[C_i(k_i(\vec{v}))] = SW(i) \geq 0$. This is without loss of generality since one could always drop goods for which $SW(i) < 0$ without any loss in social welfare. Now, we have formally defined the supply constraints ($(k^*_i)_{i \in \mathcal{I}}$) and per-good prices $((p^*_i)_{i \in \mathcal{I}})$ for our on the fly mechanism. We remark here on the integrality of $k^*_i$. In the event that $k^*_i$ is not integral, we can set the supply of good $i$ to be a random variable drawn independently according to the distribution of $k_i(\vec{v})$. In order to keep the notation from becoming messy, we prove the theorem for the deterministic supply case although the proof extends to the case where the supply is a random variable. Note that $p^*_i$ is still deterministic.

It only remains to prove that these prices and supply constraints satisfy the conditions of Lemma~\ref{lem_structural} with $\alpha=2$.
\begin{enumerate}
\item Condition (1): $p^*_ik^*_i - C_i(k^*_i) \geq 0$.

By definition, $p^*_ik^*_i - C_i(k^*_i) = \frac{V_i}{2} + \frac{1}{2} E_{\vec{v} \sim \mathcal{F}}[C_i(k_i(\vec{v})] - C_i(k^*_i)$. Since the cost functions are convex, we can apply a form of Jensen's inequality (see Lemma~\ref{app_lemma_convexexpectation}) and get that $E_{\vec{v} \sim \mathcal{F}}[C_i(k_i(\vec{v})] \geq C_i(E_{\vec{v} \sim \mathcal{F}}[k_i(\vec{v})]) = C_i(k^*_i)$. Therefore, we get that $p^*_ik^*_i - C_i(k^*_i) \geq \frac{V_i}{2} - \frac{E_{\vec{v} \sim \mathcal{F}}[C_i(k_i(\vec{v})]}{2} = \frac{1}{2}SW(i) \geq 0$.

\item Condition (2): $V_i - p^*_iE_{\vec{v} \sim \mathcal{F}}[k_i(\vec{v})] \geq \frac{1}{\alpha}SW(i)$.

As per the definition of $p^*_i$, we have that $p^*_ik^*_i = \frac{1}{2}(V_i + E_{\vec{v} \sim \mathcal{F}}[C_i(k_i(\vec{v}))])$. Taking the negative of this equation and adding $V_i$ to both sides gives us $V_i - p^*_ik^*_i = V_i - \frac{V_i}{2} - \frac{1}{2}E_{\vec{v} \sim \mathcal{F}}[C_i(k_i(\vec{v}))]) = \frac{1}{2}(V_i - E_{\vec{v} \sim \mathcal{F}}[C_i(k_i(\vec{v}))])) = \frac{1}{2}SW(i).$ 
%
%

\item Condition (3): $p^*_ik^*_i - C_i(k^*_i) \geq \frac{1}{\alpha}SW(i)$.

From the proof of Condition (1), we have that $p^*_ik^*_i - C_i(k^*_i) \geq \frac{1}{2}SW(i)$. $\qed$
\end{enumerate}

\section{Commitment Mechanisms and XoS Valuations}
\label{sec:commitment}
We move on to the much harder class of mechanisms, where the seller has to initially commit to producing a fixed amount of each good, and therefore, incurs a loss on every item that is not sold. Before showing our result, we formally define the notion of an $\alpha$-bounded solution Given an instance of our problem and an algorithm $Alg$, using the same notation as in the proof of Theorem~\ref{thm:otfbayesin}, the expected social welfare provided by the algorithm is $E_{\vec{v} \sim F}[SW(Alg(\vec{v}))] = E_{\vec{v} \sim F}[\sum_{j=1}^N v_j(A_j(\vec{v})) - \sum_{i \in \mathcal{I}}C_i(k_i(\vec{v}))]$. Then, the instance is said to be \emph{$\alpha$-bounded} with respect to $Alg$ for $\alpha \geq 1$ $if E_{\vec{v} \sim F}[\sum_{j=1}^N v_j(A_j(\vec{v}))] \geq E_{\vec{v} \sim F}[\sum_{i \in \mathcal{I}}C_i(k_i(\vec{v}))]$.

$\alpha$-boundedness captures the ratio of the social value in the computed allocation with respect to the cost of production, or equivalently, the \emph{recoverable welfare} in terms of the initial investment. In previous work~\cite{feigeIMN13}, $\alpha$-boundedness was defined in terms of the optimum solution for a given instance; on the contrary, our definition depends both on the instance and the benchmark algorithm.

%

\begin{theorem}
\label{thm:commitment}
\begin{enumerate}

\item Given an allocation algorithm $Alg$, for every instance with fractionally subadditive buyers that is $\alpha$-bounded with respect to $Alg$ for $\alpha \in [2,\infty]$, there exists a commitment mechanism that extracts a welfare of $E_{\vec{v} \sim \mathcal{F}}[\frac{SW(Alg)}{2}]\frac{\alpha-2}{\alpha-1}$
 
\item For any given instance with fractionally subadditive buyers, we can compute posted prices in poly-time such that the resulting commitment mechanisms always guarantees a social welfare within factor $4\frac{\alpha-1}{\alpha-2}$ of the optimum allocation, as long as the instance is $\alpha$-bounded with respect to the algorithm from Claim~\ref{clm_allocationalg}.
\end{enumerate}
\end{theorem}

\noindent\textbf{Discussion:} Essentially the parameter $\alpha$ allows us to distinguish between \emph{good} and \emph{bad instances} of our problem. The performance of our commitment mechanism improves as $\alpha$ increases, and for large enough values of $\alpha$, approaches that of the on the fly mechanism. In many markets, it is reasonable to expect that the cost of producing the goods is not too large in comparison to the social value or recovered welfare~\cite{feigeIMN13}. Alternatively, if the production cost function is sufficiently non-linear, then the solution returned by our algorithm from Claim~\ref{clm_allocationalg} is always $\alpha$-bounded for a large enough value of $\alpha$ to ensure good welfare. We formally prove this in the Appendix. For example, if $c_i(n) = n^2$ and $Alg$ allocates at least $3$-copies of every good, then the resulting allocation is $\alpha$-bounded for $\alpha=2.5$ and posted pricing mechanism provides a welfare guarantee of $\frac{1}{6}E_{\vec{v}}[SW(Alg))]$.

\begin{proof} The proof is somewhat similar to that of the on the fly mechanism for XoS buyers where we defined profit and surplus functions for every item and selected a price $p$ so that the profit due to the sold items is exactly the welfare due to the unsold items. However, this idea is no longer true for commitment mechanisms because even when items are not sold out, welfare may be low due to the production costs. To compensate for the negative cost when items are unsold, the main idea is to reduce prices so that the higher surplus offsets the increased production costs when the items are unsold.  Once we select the correct price per good to compensate for production loss, the rest of the proof follows almost similar to that of Theorem~\ref{thm:otfbayesin}. 
We use the same notation as we did in the proof of Theorem~\ref{thm:otfbayesin}. Suppose that $Alg$ is some given algorithm for the allocation problem and that $Alg(\vec{v}) = (A_1(\vec{v}), A_2(\vec{v}), \ldots, A_N(\vec{v}))$ denotes the output allocation for valuation profile $\vec{v}$. Moreover, suppose that under this allocation, $N_i(\vec{v})$ represents the set of buyers to whom good $i$ is allocated, and $k_i(\vec{v}) = |N_i(\vec{v})|$. We also define the following quantities for each good $i$ as we did in that proof: $V_i(\vec{v}))$ (for all $\vec{v} \sim \mathcal{F}$), $V_i$ and $SW_i$, the latter two are in expectation.

Let us begin the proof by defining a reduced surplus function for each good $i$, 

$$h_i(p) = E_{\vec{v} \sim \mathcal{F}}[V_i(\vec{v}) - pk_i(\vec{v}) - C_i(k_i(\vec{v}))] = V_i - pk^*_i - E_{\vec{v}}[C_i(k_i(\vec{v}))],$$

where $k^*_i = E_{\vec{v} \sim \mathcal{F}}[k_i(\vec{v})]$. Recall the profit function, $\pi_i(p) = pk^*_i - E_{\vec{v}}[C_i(k_i(\vec{v}))].$ For every good $i \in \mathcal{I}$, we select its price $p^*_i$ to be the point where $h_i(p) = \pi_i(p)$, i.e., where reduced surplus and profit meet. Therefore,

$$p^*_i = \frac{1}{2k^*_i}V_i.$$

For the same valuations, our commitment mechanism has a smaller price for each good than the OTF mechanism in order to limit the possibility of unsold items. Our commitment mechanism works as follows: the seller commits to producing $k^*_i$ units of good $i$ for a price of $p^*_i$ per unit. 

Using the same sequence of ideas as in the proof of Lemma~\ref{lem_structural}, we can bound the expected surplus as,

\begin{align*}
E_{\vec{v} \sim \mathcal{F}}[\sum_{j =1}^N u_j(\vec{v})] & \geq \sum_{i \in \mathcal{I}}Pr(\text{i is not sold out}) (E_{\vec{v}}[V_i(\vec{v})] - p^*_iE_{\vec{v}}[k_i(\vec{v})])\\
& = \sum_{i \in \mathcal{I}}Pr(\text{i is not sold out}) (V_i - p^*_ik^*_i).
\end{align*}

Moving on to the profit, for any fixed valuation $\vec{v}$, let $S(\vec{v})$ and $U(\vec{v})$ be the set of items that are sold out (all $k^*_i$ copies), and the ones that are not respectively. Suppose that for any $\vec{v}$, $t_i(\vec{v})$ is the number of copies of good $i$ sold by the mechanism. Then, the expected profit is at least 

\begin{align*}
&= E_{\vec{v}} [\sum_{i \in S(\vec{v})}(k^*_ip^*_i - \sum_{n=1}^{k^*_i}c_i(n)) + \sum_{i \in U(\vec{v})}t_i(\vec{v})p^*_i - \sum_{n=1}^{t_i(\vec{v})}c_i(n))] \\
& \geq \sum_{i \in \mathcal{I}}\{Pr(\text{i is sold out})(p^*_ik^*_i - C_i(k^*_i)) - Pr(\text{i is not sold out})C_i(k^*_i)\}.
\end{align*}

Adding the profit and surplus gives a lower bound on the expected welfare due to our mechanism. Also recall that due to our choice of $p^*i, k^*_i$, the following is true $V_i - p^*_ik^*_i - C_i(k^*_i) = p^*_ik^*_i - C_i(k^*_i).$ So, we get that the expected social welfare of our mechanism $\mathcal{M}$ is at least,

\begin{align*}
E[SW(\mathcal{M}] & \geq \sum_{i \in \mathcal{I}}Pr(\text{i is not sold out}) (V_i - p^*_ik^*_i - C_i(k^*_i)) + \sum_{i \in \mathcal{I}}Pr(\text{i is sold out})(p^*_ik^*_i - C_i(k^*_i))\\
& = \sum_{i \in \mathcal{I}}Pr(\text{i is not sold out}) (p^*_ik^*_i - C_i(k^*_i)) + \sum_{i \in \mathcal{I}}Pr(\text{i is sold out})(p^*_ik^*_i - C_i(k^*_i)) \\
& = \sum_{i \in \mathcal{I}}(p^*_ik^*_i - C_i(k^*_i)).
\end{align*}

How does this compare to the $E_{\vec{v}}[SW(Alg)] = \sum_{i \in \mathcal{I}}(V_i - E_{\vec{v}}[C_i(k_i(\vec{v}))]$? Observe that by Jensen's inequality (Lemma~\ref{app_lemma_convexexpectation}), we know that $E_{\vec{v}}[C_i(k_i(\vec{v}))] \geq C_i(k^*_i)$ and therefore, $E_{\vec{v}}[SW(Alg)] \leq \sum_{i \in \mathcal{I}}(V_i - C_i(k^*_i))$. The ratio of welfare of $Alg$ with respect to that of our mechanism is therefore at most,

$$\frac{\sum_{i \in \mathcal{I}}(V_i - C_i(k^*_i))}{\sum_{i \in \mathcal{I}}(k^*_ip^*_i - C_i(k^*_i))} = \frac{\sum_{i \in \mathcal{I}}(V_i - C_i(k^*_i))}{\sum_{i \in \mathcal{I}}(\frac{V_i}{2} - C_i(k^*_i))}.$$

The second expression comes from the fact that $p^*_i = \frac{1}{2k^*_i}.$

Since the solution provided by the algorithm is $\alpha$-bounded, we know that $\sum_{i}V_i \geq \alpha \sum_i E_{\vec{v}}[C_i(k_i(\vec{v}))] \geq \alpha \sum_i C_i(k^*_i)$ for some $\alpha \in [2,\infty]$. Therefore, the ratio of the expected welfare is at most $\frac{\sum_{i \in \mathcal{I}}(\alpha C_i(k^*_i) - C_i(k^*_i))}{\sum_{i \in \mathcal{I}}(\frac{\alpha}{2}C_i(k^*_i) - C_i(k^*_i))}$.

After some elementary algebra, we get that 

$$\frac{E_{\vec{v}}[SW(Alg)]}{E_{\vec{v}}[SW(\mathcal{M}] }\leq 2\frac{\alpha-1}{\alpha-2}.$$

The second part of the theorem follows directly from the algorithm in Claim~\ref{clm_allocationalg}.


\end{proof}


\section{Generalizations}
\label{sec:generalizations}
In previous section, we presented a $O(1)$-approximate on the fly mechanism and a $O(\frac{\alpha-1}{\alpha-2})$-approximate commitment mechanism for XoS buyers. Restricting our focus to OTF mechanisms, we now generalize Theorem~\ref{thm:otfbayesin} in two directions. First, we consider markets where buyer valuations are subadditive (strictly more general than fractionally subadditive) and present mechanisms that extract a logarithmic fraction of the optimum welfare for the restricted case of limited supply. Following this, we move on to the more daunting problem of designing good mechanisms when the seller is completely oblivious to the buyer valuations. Leveraging our black-box existence result of Theorem~\ref{thm:otfbayesin}, we use `guess prices' to design mechanisms with reasonable performance for XoS buyers and convex production costs. 

\subsection{Subadditive Buyers and Limited Supply: Single Price Mechanisms}
\label{sec:subadditive}
We now move on to setttings where buyers have subadditive valuation functions, which are a strict generalization of fractionally subadditive functions as well as the largest class of complement-free valuations. Our main result in this section is a posted price mechanism\footnote{Recall that there is no difference between on the fly and commitment mechanisms when we have limited supply.} for instances where the seller has a limited supply  of each item; this mechanism posts a single price per good and achieves a $O(\log(M))$-approximation to the optimum social welfare. 

In~\cite{feldmanGL15}, the authors remarked that their mechanism could be used to obtain a logarithmic approximation to welfare for unit supply settings using the fact that every subadditive function can be approximated by an XoS function with a logarithmic distortion (see~\cite{bhawalkarR11} for a precise definition). However, it is unclear whether their techniques imply a computationally efficient mechanism to this effect. We show how to efficiently compute posted prices that achieve the desired approximation guarantee even when the seller has an arbitrary supply of various items. All of the proofs from this section are present in the Appendix owing to their similarities to our previous proofs.

%

\begin{theorem}
\label{thm_subadditive}
\begin{enumerate}
\item (Black-box reduction) For any instance with subadditive buyer valuations and arbitrary supply, given an allocation algorithm $Alg$, we can compute a price $p_i$ for every good $i$ such that an on the fly mechanism with price vector $\vec{p} = (p_1, \ldots, p_m)$ and the original supply constraints yields a welfare of $\frac{E_{\vec{v} \sim \mathcal{F}}[SW(Alg(\vec{v})]}{O(\log M)}$. 

\item (Computational Result) For any instance with subadditive buyer valuations and arbitrary supply, we can compute posted prices in poly-time such that the resulting mechanism is an $O(\log(M))$-approximation to the optimum social welfare.

\end{enumerate}

\end{theorem}

The computational result is achieved by plugging in the $2$-approximation algorithm for allocating goods to subadditive buyers by Feige~\cite{feige09}. Notice that while our previous results for XoS buyers hold for arbitrary convex cost functions, the mechanism for subadditive buyers achieves a logarithmic approximation only for the special case with limited supply. In fact, we believe that the black-box reduction extends naturally to settings with production. However, the bottleneck that prevents us from presenting an actual mechanism is the absence of a constant factor approximation algorithm for the allocation problem for subadditive buyers with production cost functions. An interesting open question is whether Feige's celebrated result can be extended to settings with convex costs.

%
%
%
%

\subsection{Unknown buyer distributions: $O(\log N \log M)$-approximate OTF Mechanisms}
\label{sec:unknownbuyers}
All of the mechanisms that we have designed so far strongly exploit the Bayesian environment and the prior distributions to derive item prices. In an attempt to wean our mechanisms off the strong dependence on the valuation distributions, we consider the problem of computing prices for on the fly mechanisms when the seller has absolutely no information regarding buyer valuations. The non-discriminatory nature of our mechanisms prevent
the use of standard learning techniques to approximate the buyer distributions followed by an application of Theorem~\ref{thm:otfbayesin}. Instead, we 
present a simple randomized mechanism that provides a logarithmic approximation to optimum welfare.

\begin{theorem}
\label{thm_unknown}
For any setting with convex production costs where buyer valuations belong to the class of XoS functions, we can design a computationally efficient on the fly mechanism that results in a $O(\log(N) \log(M))$-approximation to the optimum social welfare.
\end{theorem}

The above mechanism does not require access to any type of oracle. However, we assume that the seller is capable of efficiently estimating the social welfare of the optimum allocation (or an upper bound thereof). Actually, our results hold even when the seller can estimate only an upper bound on the optimum welfare $\hat{SW}(OPT)$ using standard techniques such as those in~\cite{chakrabortyHK09}.Moreover, unlike all of our previous results, there is no dependence on querying oracles and therefore, the runtime is strictly polynomial and there is no additive error in the welfare guarantee even when the (unknown) distributions are exponential in size.

\begin{proof}
Suppose that for a given instance, $OPT = (O_1, O_2, \ldots, O_N)$ is the welfare maximizing allocation. Let us treat the given instance as a full information problem where the seller is aware of all buyer valuations. Then, applying Theorem~\ref{thm:otfbayesin} for the given instance with $OPT$ as the solution returned by the benchmark algorithm, we obtain an on the fly mechanism with prices $(p_i)_{i \in \mathcal{I}}$ and supply constraints $(k_i)_{i \in \mathcal{I}}$ with a welfare guarantee of $\frac{SW(OPT)}{2}$. Define $\pi^*_i := p_ik_i - C_i(k_i)$ to be the profit that the mechanism makes from good $i$ when it is sold out, and $SW^*_i$ to be the total welfare due to good $i$ as we did in the proof of Theorem~\ref{thm:otfbayesin}. Without loss of generality, we make the following two assumptions regarding the structure of the optimum allocation and our OTF mechanism derived via $OPT$.

\begin{enumerate}
\item For all $i \in \mathcal{I}$, the number of units of good $i$ allocated in the optimum allocation is a power of two, and therefore, $k_i$ has the form $2^{l}$ for some $l \in [0,\log(N)]$. If this is not true, one can easily replace $OPT$ with the welfare maximizing solution where the quantity of each good allocated is a power of two (or no good is allocated at all), resulting in at most a factor two loss in social welfare.

\item Similarly, we also assume that for every good $i$, the total profit obtained by the mechanism when this good is sold out is at least $\frac{SW(OPT)}{4M}$. One can always discard goods whose welfare contribution is extremely low. For instance, in the benchmark OTF mechanism, consider all the goods $i$ satisfying $p_ik_i - C_i(k_i) < \frac{SW(OPT)}{4M}$. We know from the proof of Theorem~\ref{thm:otfbayesin} that $\sum_{i \in \mathcal{I}}p_ik_i - C_i(k_i) = \frac{SW(OPT)}{2}$. Since there are at most $M$ goods, discarding the ones with welfare smaller than $\frac{SW(OPT)}{4M}$, reduces the mechanism's social welfare by at a most a factor two.

\end{enumerate}
%

Of course, the OTF mechanism derived from $OPT$ is simply a theoretical benchmark that we will compare our actual randomized mechanism to. Since, we are not aware of the prices or the supply constraints, we will attempt to design a randomized mechanism that selects these parameters from a feasible set and show that with a probability $\frac{1}{\Omega(\log M \log N)}$, the randomized mechanism correctly guesses these parameters for each good. Then, in this case, the welfare achieved per good is close to that in the benchmark OTF mechanism. When the guessed price is incorrect, we sill prove that the social welfare due to each good $i$ is non-negative and therefore, this does not adversely effect the welfare of the mechanism.

We now formally define our mechanism. 
\begin{enumerate}
\item For each good $i \in \mathcal{I}$, select an integer $r_1$ uniformly at random from the set $[0,1,2,\ldots \log(N)]$ and set $\tilde{k}_i = 2^{r_1}$.

\item Once $\tilde{k}_i$ is selected, choose an integer $r_2$ uniformly at random from the set $[0,1,2,\ldots, 2+\log(M)]$, and set the price $\tilde{p}_i = \frac{C_i(\tilde{k}_i)}{\tilde{k}_i} + \frac{SW(OPT)}{4M\tilde{k}_i}2^{r_2}$.
\end{enumerate}

Consider the OTF mechanism with price $\tilde{p}_i$ on good $i$ and supply constraint $\tilde{k}_i$. Given a vector of prices and supply parameters, $\vec{p'}, \vec{k'}$, define $\pi_i(\vec{p'}, \vec{k'})$ to be the profit that our actual OTF mechanism obtains for good $i$ given a particularly instantiation of prices and supply, and $SW(\vec{p'}, \vec{k'})$ to be the total social welfare for the same parameters. Our first lemma establishes that for every possible instantiation of prices and supply, the profit due to any given good is non-negative.

\begin{lemma}
\label{applem_unknownprofit}
For all $i \in \mathcal{I}$, and for every possible supply vector $\vec{k'}$ and price vector $\vec{p'}$ for the mechanism, $\pi_i(\vec{p'}, \vec{k'}) \geq 0$.
\end{lemma}
\begin{proof}
Suppose that for a given choice of supply and price, $t_i$ units of good $i$ are sold.  Then, the profit due to this good is $p'_it_i - C_i(t_i)$. By definition, we know that (given $k'_i$), $p'_i \geq \frac{C_i(k'_i)}{k'_i}$ and that $t_i \leq k'_i$ since the mechanism cannot sell more than $k'_i$ units of this good. Therefore, by convexity, it is the case that $p'_i \geq \frac{C_i(t_i)}{t_i}$. Substituting this, we get $p'_it_i - C_i(t_i) \geq 0$. 
\end{proof}

Our next claim establishes the existence of an approximately `correct guess price'. Specifically, for each good $i$, we show that when $\tilde{k}_i = k_i$, i.e., the randomized mechanism correctly guesses the supply for this good (w.r.t the benchmark mechanism), $\exists$ some price $\tilde{p}_i \leq p_i$, at which $\tilde{p}_ik_i - C_i(k_i) \geq \frac{\pi^*_i}{2}$.

\begin{lemma}
\label{applem_correctguessprice}
There exists $r_2 \in [0,2+\log(M)]$ such that for $\tilde{p}_i = \frac{C_i(k_i)}{k_i} + \frac{SW(OPT)}{4Mk_i}2^{r_2}$, 
$$\tilde{p}_ik_i - C_i(k_i) \geq \frac{\pi^*_i}{2} \geq \frac{SW^*_i}{4}$$
\end{lemma}
\begin{proof}
The proof is not particularly hard to see. First observe that when $r_2 = 0$, $\tilde{p}_ik_i - C_i(k_i) = \frac{SW(OPT)}{4M} \leq \pi^*_i$ by definition. Next, when $r_2 = 2+\log(M)$, $\tilde{p}_ik_i - C_i(k_i) = SW(OPT) \geq \pi^*_i$. 

Define $g(r_2) = \{\frac{C_i(k_i)}{k_i} + \frac{SW(OPT)}{4Mk_i}2^{r_2}\}k_i - C_i(k_i)$. For some $t \in [0,1+\log(M)]$, consider $g(t)$, and $g(t+1)$. Using some basic algebra, we get that $g(t+1) = 2g(t)$. That is, for successive values of $r_2$, $g(r_2)$ is doubled. We have already shown that $g(0) \leq \pi^*_i \leq g(2+\log(M))$. Therefore, $\exists$ some $r_2$ in the range $[0,2+\log(M)]$ such that $\frac{\pi^*_i}{2} \leq g(r_2) \leq \pi^*_i$. The connection to $SW^*_i$ comes from the fact that $\pi^*_i \geq \frac{SW^*_i}{2}$, which is evident in the proof of Theorem~\ref{thm:otfbayesin}.
\end{proof}

For each good $i$, we use $p^g_i$ to denote this so-called `correct guess price'. Additionally, since the function $p'_ik_i - C_i(k_i)$ is monotone increasing in $p'_i$, the correct guess price $p^g_i \leq p_i$, the price of good $i$ in the benchmark OTF mechanism.

For any given instantiation of prices and supply for our randomized mechanism $\vec{p'}, \vec{k'}$, define $CG(\vec{p'}, \vec{k'})$ to be the set of goods for which supply equals its value in the benchmark mechanism and the price is the `correct guess price'. i.e., $CG(\vec{p'}, \vec{k'}) := \{i ~ | ~ k'_i = k_i \text{ and } p'_i = p^g_i\}$. Now, we are ready to obtain a lower bound on $SW(\vec{p'}, \vec{k'})$ in terms of the goods in $CG(\vec{p'}, \vec{k'})$. 

\begin{lemma}
\label{applem_swlb}
For any given instantiation of our algorithm, i.e., for any choice of $\vec{p'}$ and $\vec{k'}$, the resulting social welfare is given by,

$$SW(\vec{p'}, \vec{k'}) \geq \sum_{i \in CG(\vec{p'}, \vec{k'})}p^g_ik_i - C_i(k_i) \geq \sum_{i \in CG(\vec{p'}, \vec{k'})}\frac{SW^*_i}{4}.$$
\end{lemma}
\begin{proof}
Suppose that for the given choice of prices and supply, $S$ and $U$ denote the sold out and unsold goods respectively upon running the mechanism. Since we already know that the profit due to the unsold goods is non-negative (Lemma~\ref{applem_unknownprofit}), the overall profit due to the mechanism is at least the profit due to the `sold out' goods $\sum_{i \in S}p'_ik'_i - C_i(k'_i)$, which in turn is at least $\sum_{i \in S \cap CG(\vec{p'}, \vec{k'})}p^g_i k_i - C_i(k_i)$ as per the definition of $CG(\vec{p'}, \vec{k'})$. Applying Lemma~\ref{applem_correctguessprice}, we get the following final lower bound on the profit,

$$\sum_{i \in S \cap CG(\vec{p'}, \vec{k'})}\frac{1}{2}\pi^*_i \geq \sum_{i \in S \cap CG(\vec{p'}, \vec{k'}))}\frac{1}{4}SW^*_i.$$

Under the given set of prices, we know that every buyer's utility (surplus) is non-negative. So, we can bound the overall surplus (buyer utility) due to the mechanism by counting only the utility that buyers receive from the goods in $U \cap CG(\vec{p'}, \vec{k'})$. Since these goods are not sold out, they are clearly available to be purchased by each buyer. Moreover, for each buyer $i$, it suffices to consider the utility that this buyer would have received from the subset of goods in $O_i$ that are available in $U \cap CG(\vec{p'}, \vec{k'})$. Suppose that for each $i \in \mathcal{I}$, $V^*_i$ is the total value from good $i$ in the optimum solution (similar to our definitions in the proof of Theorem~\ref{thm:otfbayesin}).

\begin{align*}
\text{Surplus} & \geq \sum_{j=1}^N v_j(O_j \cap (U \cap CG(\vec{p'}, \vec{k'}))) - \sum_{i \in O_j \cap (U \cap CG(\vec{p'}, \vec{k'}))}p'_i\\
& \geq \sum_{i \in U \cap CG(\vec{p'}, \vec{k'}))}[V^*_i - p^g_ik_i]. \\
& \geq \sum_{i \in U \cap CG(\vec{p'}, \vec{k'}))}\frac{1}{2}SW^*_i.
\end{align*}

The last inequality comes from Lemma~\ref{lem_structural} and the fact that $p^g_i \leq p_i$. Adding the profit and surplus yields the lemma.

\end{proof}

We can now bound the overall social welfare using just the social welfare due to each good given that the guess price is correct, since we know that the social welfare is non-negative otherwise. Formally, using Lemma~\ref{applem_swlb}

\begin{align*}
E[SW] = &  \sum_{(\vec{p'}, \vec{k'})}Pr(\vec{p'}, \vec{k'}) SW(\vec{p'}, \vec{k'})\\
\geq &  \sum_{(\vec{p'}, \vec{k'})}Pr(\vec{p'}, \vec{k'}) \sum_{i \in CG(\vec{p'}, \vec{k'})}\frac{SW^*_i}{4} \\
= & \sum_{i \in \mathcal{I}}Pr(p'_i =p^g_i \text{ and } k'_i = k_i)\frac{SW^*_i}{4} \\
= & \sum_{i \in \mathcal{I}}\frac{1}{(2+\log M)(1+\log N)}\frac{SW^*_i}{4} \\
= & \frac{1}{4(2+\log M)(1+\log N)}SW(OPT)
\end{align*}

This completes the proof.%
\end{proof}

\noindent\textbf{Acknowledgements} The author wishes to express his gratitude to Elliot Anshelevich and Brendan Lucier for taking the time to provide useful feedback on the manuscript.
\bibliography{bibliography}
\bibliographystyle{plain}

\newpage

\appendix

\section{Additional Lemmas}

\begin{lemma}
	\label{app_lemma_convexexpectation}Let $C$ be a convex function and let $x$ be a random variable distributed according to $F$. Then, $E_{x \sim F}[C(x)] \geq C(E_{x \sim F}[x])$.
\end{lemma}
\begin{proof}
This follows from the fundamental property of convex functions that $C(\frac{x+y}{2}) \leq \frac{1}{2}(C(x) + C(y))$, also referred to as Jensen's Inequality.
\end{proof}

\section{Proofs from Section~\ref{sec:framework}}
\begin{lem_app}{lem_posprofit}
For all $t \leq k$, the profit that the seller makes by selling $t$ copies of good $i$ at price $p^*_i$ is at least the cost of producing these $t$ copies, i.e., $t p^*_i \geq \sum_{n=1}^t c_i(n)$.
\end{lem_app}

\begin{proof}
From the previous lemma, we know that $\sum_{n=1}^k [p^*_i - c_i(n)]$ is exactly half the social welfare due to good $i$. Since we assumed that the overall welfare is non-negative, this means that $\sum_{n=1}^k [p^*_i - c_i(n)] \geq 0$. Now, remember that the cost function is convex, therefore, $p^*_i - c_i(n_1) \geq p^*_i - c_i(n_2)$ for $n_1 \leq n_2$. 

Suppose that by contradiction, the lemma statement is false. So, $\sum_{n=1}^t [p^*_i - c_i(n)] < 0$. By the convexity of the production cost function, it has to be true that $p^*_i - c_i(t) < 0$, and therefore for all $n > t$, $p^*_i - c_i(n) < 0$. Summing up over all $n > t$, we get $\sum_{n=t+1}^k[p^*_i - c_i(n)] < 0$. By our initial assumption, $\sum_{n=1}^t[p^*_i - c_i(n)] < 0$, giving us a total negative profit which contradicts Lemma~\ref{lem_revsocwel}. 
\end{proof}

\subsection{Combinatorial $2$-approximation algorithm for XoS Valuations}

We now present a simple, combinatorial $2$-approximation algorithm for the allocation problem with XoS buyers and production costs. This result extends the $2$-approximation algorithm described in~\cite{dobzinskiNS10} for unit-supply to more general settings with production costs. When used in conjunction with our central existence result (Theorem~\ref{thm:otfbayesin}), the current algorithm yields a set of non-discriminatory posted prices (per good) that achieve one fourth of the optimum social welfare. 

Formally, we are given a set of goods $\mathcal{I}$ with convex production cost function $C_j()$ for each good $j$ and the corresponding (non-decreasing) marginal cost function $c_j$. There are set of $N$ buyers $\{1,2,\ldots, N\}$ with each buyer having a fractionally subaddtive valuation function $v_i()$. Finally, we assume that the following algorithm has access to both a demand oracle and a XoS oracle. 

\subsection{Combinatorial Algorithm}
\begin{itemize}
\item Initialization
\begin{enumerate}
\item Let $S_0 = \mathcal{I}$ and $S_1 = S_2 = \ldots = S_N = \emptyset$ denote the initial allocation to buyers
\item For all $j \in \mathcal{I}$, set its initial price $p_j = c_j(1)$.
\item For all $j \in \mathcal{I}$, let $X_j = \{0\}$ denote the set of buyers to whom good $j$ has been allocated.
\item For all buyers $i$, goods $j$, let $q_i(j) = 0$ be the price assigned to item $j$ by buyer $i$.
\item For all goods $j$, let $y_j = 0$ denote the buyer offering the good at the cheapest price.
\end{enumerate}

\item Allocation: For each buyer $i \in 1,2, \ldots, N$
\begin{enumerate}

\item Let $S'_i$ be the max-utility bundle demanded by buyer $i$ at prices $\vec{p}$; set $S_i = S'_i$
\item Let $(x^i_j)_{j \in S'_i}$ denote the maximizing XoS clause for buyer $i$ and set $S'_i$. 
\item Updation: For all $j \in S'_i$,

\begin{itemize}
\item If $y_j \neq 0$, $S_{y_j} = S_{y_j} \setminus \{j\}$, $X_j = X_j \cup \{i\} \setminus \{y_j\}$.
\item $q_i(j) = x^i_j$ and $q_0(j) = c_j(|X_j|)$. 
\item Set $p_j = \min_{i \in X_j} q_i(j)$; set $x_j = \argmin_{i \in X_j} q_i(j)$.

\end{itemize}

\end{enumerate}
\end{itemize}

\begin{clm_app}{clm_allocationalg} ($2$-Approximation Algorithm) Given a valuation profile $\vec{v} = (v_1, v_2, \ldots, v_N)$ of XoS functions, and access to demand, and XoS oracles, we can compute in poly-time an allocation of goods to the buyers whose social welfare is at least half of that of the welfare maximizing allocation.

\end{clm_app}
\begin{proof}
\subsubsection*{Informal Description of Algorithm}
We now describe the various parts of the algorithms and provide intuition for its working. In addition to the $N$ buyers, we also add a buyer $0$ (dummy buyer), whose role will be akin to the seller selling the goods at their marginal production cost. Initially, suppose that buyer $0$ produces one copy of all the goods and prices them at $c_j(1)$ for all $j \in \mathcal{I}$. Now, in each round, exactly one buyer \textit{arrives} (this is indeed analogous to our sequential pricing mechanism), observes the prices of all the goods in $\mathcal{I}$ and purchases a bundle that maximizes her utility; the utility-maximizing bundle is computed using a demand oracle. This buyer (say $i$) prices her bundle ($S'_i$) using the additive clause that maximizes the XoS valuation for $S'_i$. One can imagine that in each round multiple copies of a good are being sold by different buyers at different prices. However, for obvious reasons, we are only concerned with the buyer selling this good at the minimum price $p_j$. Finally, at the end of each round, buyer $0$ replenishes the goods that have been purchased from him by producing a copy and increasing their price to the current marginal price. 

We observe that buyer $i$ starts out with the allocation $S'_i$, pricing each $j \in S'_i$ at $q_i(j)$. Some of these goods are sold off to future buyers and in every successive round, buyer $i$'s allocation is a subset of $S'_i$. Finally, buyer $i$ never changes the price of the goods she sells.

\subsection*{Analysis of Algorithm}

\textbf{Notation.} Naturally, we assume that the algorithm proceeds in rounds wherein buyer $i$ is allocated some set of goods in round $i$ (some of these goods may later be reallocated to later buyers). As mentioned in the algorithm, let $S'_i$ denote the set first allocated to buyer $i$ (in round $i$), $S^{(t)}_i \subseteq S'_i$ are the goods still present in the hands on buyer $i$ at the end of round $t$ and $S_i$ is the final allocation to this buyer, i.e., $S_i = S^{(N)}_i$. Next, suppose that $p^{(t)}_j$ denotes the price of good $j$ at the beginning of round $r$, and $q_i(j)$ is the value of good $j$ in the XoS clause of buyer $i$ corresponding to $S'_i$ (as mentioned in the algorithm). Finally, let $X_j$ and $X^*_j$ denote the set of buyers to whom good $j$ is allocated by our algorithm and in the optimum solution respectively, and $x^*_j = |X^*_j|$, and $x_j = |X_j|$. Similarly, let $X^{(r)}_j$ denote the set of buyers to whom item $j$ is assigned at the end of round $r$ (not including the zero buyer), and $x^{(r)}_j = |X^{(r)}_j|$.

We begin by stating a few simple observations regarding our algorithm.

\begin{claim}
\label{prop_combxosobservations}
\begin{enumerate}
\item The price of any good is non-decreasing during the course of the algorithm.

\item The set of goods allocated to each buyer is monotonically non-increasing.

\item For any good $j$, round $r$ and buyer $i \in X_j$, $p^{(r)}_j \leq q_i(j)$. 

\item For any good $j$, suppose that $x_j$ units of this good are sold in total by our algorithm. Then, for any $r \leq N$, we have that $p^{(r)}_j \leq c_j(x_j+1)$. 

\end{enumerate}
\end{claim}

For any set $T$ of goods, round $r$, we abuse notation and define $p^{(r)}(T) := \sum_{j \in T}p^{(r)}(j)$. Suppose that $(O_1, O_2, \ldots, O_N)$ denotes the optimum allocation to the buyers. 

\begin{lemma}
\label{lem_utility}
$$\sum_{i=1}^N [v_i(O_i) - p^{(i)}(O_i)] \leq \sum_{i=1}^{N}v_i(S_i) - \sum_{j \in \mathcal{I}}C_j(x_j).$$ 
\end{lemma}

\begin{proof}
We prove a slightly stronger claim by induction. For every $1 \leq k \leq N$, we claim that

$$\sum_{i=1}^k [v_i(O_i) - p^{(i)}(O_i)] \leq \sum_{i=1}^{k} \sum_{j \in S^{(k)}_i} q_i(j) - \sum_{j \in \mathcal{I}}C_j(x^{(k)}_j).$$ 
The base case for $k=1$ is rather straightforward. Recall that for every $j \in \mathcal{I}$, $p^{(1)}_j = c_j(1)$. Since buyer $1$ purchased the set $S'_1$ during round $1$ (it is her utility maximizer at these prices), we have that $v_1(O_1) - p^{(1)}(O_1) \leq v_1(S'_1) - p^{(1)}(S'_1) \leq \sum_{j \in S'_1}q_1(j) - \sum_{j \in S'_1}c_j(1).$

Suppose that the claim holds up to $k-1$. Recall that $S'_k$ is the bundled purchased by buyer $k$ under the prices $\vec{p}^{(k)}$. Suppose that we divide $S'_k$ into two subsets, $T_1$ and $T_2$ such that $T_1$ refers to the items for which $y_j = 0$ at the beginning of round $k$, i.e., buyer $k$ steals these goods from the dummy buyer while $T_2 := S'_k \setminus T_2$ and this represents the goods stolen from actual buyers. For every $j \in T_1$, we have that $p^{(k)}(j) = c_j(x^{(k)}_j + 1)$. Now, for $j \in T_2$, let $i_j := \argmax_{i' \in X^{(k)}_j}q_{i'}(j)$. Then, we have that $p^{(k)}(j) = q_{i_j}(j)$. 

Therefore, the quantity $p^{(k)}(S'_k) = \sum_{j \in T_1}c_j(x^{(k)}_j + 1) + \sum_{j \in T_2} q_{i_j}(j)$. We are now ready to show the base claim. Observe that since buyer $k$ chose her best bundle in round $k$, $v_k(O_k) - p^{(k)}(O_k) \leq v_k(S'_k) - p^{(k)}(S'_k) = \sum_{j \in S'_k}q_i(j) -  \sum_{j \in T_1}c_j(x^{(k)}_j + 1) - \sum_{j \in T_2} q_{i_j}(j)$.

\begin{align*}
\sum_{i=1}^k [v_i(O_i) - p^{(i)}(O_i)] & = \sum_{i=1}^{k-1} \sum_{j \in S^{(k-1)}_i}q_i(j)  - \sum_{j \in \mathcal{I}}C_j(x^{(k-1)}_j) + \sum_{j \in S'_k}q_i(j) -  \sum_{j \in T_1}c_j(x^{(k)}_j + 1) - \sum_{j \in T_2} q_{i_j}(j) \\
& \leq \sum_{i=1}^{k-1} \sum_{j \in S^{(k-1)}_i}q_i(j) - \sum_{j \in \mathcal{I}}C_j(x^{(k)}_j) + \sum_{j \in S'_k}q_i(j) - \sum_{j \in T_2} q_{i_j}(j) \\
& = \sum_{i=1}^{k} \sum_{j \in S^{(k)}_i} q_i(j) - \sum_{j \in \mathcal{I}}C_j(x^{(k)}_j).
\end{align*}

We now explain the last line of the above proof sequence. For every good $j \in T_2$, the term $q_{i_j}(j)$ appears as a positive term once in the beginning since $j \in S^{(k-1)}_{i_j}$ and once as a negative term, and both of these terms cancel out. 
This completes the proof of the lemma.
\end{proof}	

\begin{lemma}
\label{lem_poi}
\end{lemma}
$$\sum_{i=1}^N p^{(i)}(O_i) \leq C(OPT) + \sum_{i=1}^N v_i(S_i) - \sum_{j \in \mathcal{I}}C_j(x_j).$$
\begin{proof}
The inequality in the lemma statement can be expanded as

$$\sum_{i=1}^N \sum_{j \in O_i}p^{(i)}(j) \leq \sum_{j \in \mathcal{I}}C_j(x^*_j) + \sum_{i=1}^N v_i(S_i) - \sum_{j \in \mathcal{I}}C_j(x_j).$$

We now prove that for any $j \in \mathcal{I}$, $\sum_{i \in X^*_j}p^{(i)}(j) \leq C_j(x^*_j) + \sum_{i \in X_j}q_i(j) - C_j(x_j)$.

Summing this up over all $j$, and using the fact that for any buyer i, $\sum_{j \in S_i}q_i(j) \leq v_i(S_i)$ gives us the lemma. Now, we prove the sub-claim in two cases.

\textbf{Case I}: $x _j < x^*_j$. 

That is, buyers consume more of good $j$ in the optimum solution than in that returned by our algorithm. Arbitrarily partition $X^*_j$ into two sets of buyers $X_1$ and $X_2$ with the only condition being $|X_1| = |X_j|$. Moreover, as per Claim~\ref{prop_combxosobservations} (statement 4), for every $i \in X^*_j$, $p^{(i)}_j \leq c_k(x_j+1)$. So, we have that

$$\sum_{i \in X^*_j}p^{(i)}(j) \leq \sum_{i \in X_1}p^{(i)}(j) + \sum_{i \in X_2}p^{(i)}(j) \leq \sum_{i \in X_j}q_i(j) + \sum_{r=x_j+1}^{x^*_j}c_j(r).$$ 

Adding and subtracting $C_j(x_j)$, we get the desired inequality.

$$\sum_{i \in X_j}q_i(j) + \sum_{r=x_j+1}^{x^*_j}c_j(r) + C_j(x_j) - C_j(x_j) \leq C_j(x^*_j) + \sum_{i \in X_j}q_i(j) - C_j(x_j).$$

\textbf{Case II}: $x^*_j \leq x_j$.

We make a small observation before proving this case. Let $i$ be any buyer belonging to $X_j$, then we claim that $q_i(j) \geq c_j(x_j)$. This is a rather fundamental requirement for any algorithm attempting to maximize social welfare. To see why this is true, observe that since $x_j$ units of the good are sold, at some round $r \leq N$, the price of good $j$ must have been $c_j(x_j)$. Moreover, since the price of the good is non-decreasing, $p^{(N)}(j) \geq c_j(x_j)$. Now, applying Claim~\ref{prop_combxosobservations}, we get the desired result. 

Moving on, this time we partition $X_j$ into two subsets $X_1$ and $X_2$ with $|X_1| = |X^*_j|$. Therefore, we have that

$$\sum_{i \in X^*_j}p^{(i)}(j) \leq \sum_{i \in X_1}q_i(j) + C_j(x_j) - C_j(x_j) = C_j(x^*_j) + \sum_{i \in X_1}q_i(j) + \sum_{r=x^*_j+1}^{x_j}c_j(r) - C_j(x_j).$$

Using the observation that for all $r \leq x_j$, $c_j(r) \leq q_i(j)$ for any $i \in X_2$, allows us to complete the claim.

\end{proof}

\end{proof}

\section{Proofs from Section~\ref{sec:commitment}}
\subsection{Convex Production Costs lead to Good Commitment Mechanisms}
Before our proving main result from Section~\ref{sec:commitment}, we take a brief digression to establish sufficient conditions for an $\alpha$-bounded instance to have a large enough value of $\alpha \geq 2$ purely in terms of the convexity of the production cost function. Let us first define the notion of a $\gamma$-convex cost function. 

\begin{definition}
A production cost function $C$ is said to be $\gamma$-convex for some $\gamma \in [1,\infty)$ if for all $k \geq 3$, $\gamma C_i(k) \leq k c_i(k)$. 
\end{definition} 

We observe that every convex cost function is $\gamma$-convex for some value of $\gamma$ in the given range. A larger value of $\gamma$ implies greater convexity. As an example, consider the function $c_i(n) = n^2$. It is not hard to show that this function is $\gamma$-convex for $\gamma=2.5$. We now establish a relationship between $\gamma$-convexity and $\alpha$-bounded instances. 

\begin{claim}
Consider a market with XoS buyers where the production cost functions on all of the goods are $\gamma$-convex. Given an algorithm $Alg$, and an instance for which the algorithm allocates at least $3$ copies of every good to the buyers. Then, the instance is $\alpha$-convex with respect to $Alg$ for $\alpha=\gamma$.
\end{claim}
\begin{proof}
As usual, we assume that the social welfare of the solution returned by the algorithm cannot improved by un-allocating a good to some buyer. Therefore, for a deterministic allocation $(A_1, \ldots, A_N)$, let $o_j$ be the maximizing XoS clause with respect to $A_j$. Then, for every $i \in A_j$, $o_j(i)) \geq c_i(k_i)$, where $k_i)$ is the number of copies of $i$ allocated to buyers.

Therefore, we have that the total value derived by buyers is

$$\sum_{j=1}^N \sum_{i \in \mathcal{A_j}}o_j(i) \geq \sum_{j=1}^N \sum_{i \in A_j}c_i(k_i) = \sum_{i \in \mathcal{I}}k_ic_i(k_i) \geq \gamma \sum_{i \in \mathcal{I}}C_i(k_i). $$
\end{proof}

\section{Proofs from Section~\ref{sec:subadditive}}

\begin{thm_app}{thm_subadditive}
\begin{enumerate}
\item (Black-box reduction) For any instance with subadditive buyer valuations and arbitrary supply, given an allocation algorithm $Alg$, we can compute a price $p_i$ for every good $i$ such that an on the fly mechanism with price vector $\vec{p} = (p_1, \ldots, p_m)$ and the original supply constraints yields a welfare of $\frac{E_{\vec{v} \sim \mathcal{F}}[SW(Alg(\vec{v})]}{O(\log M)}$. 

\item (Computational Result) For any instance with subadditive buyer valuations and arbitrary supply, we can compute posted prices in poly-time such that the resulting mechanism is an $O(\log(M))$-approximation to the optimum social welfare.

\end{enumerate}

\end{thm_app}
\begin{proof}
The proof is rather similar to that of Theorem~\ref{thm:otfbayesin} with one crucial difference: for subadditive buyers, we can no longer directly break an allocation down into its constituent per-item XoS prices (obtained via the additive clauses). Therefore, even for a fixed valuation profile $\vec{v} \sim \mathcal{F}$, given the allocation returned by the algorithm, it is unclear how to derive per-good prices. Our following structural lemma inspired by a similar claim in~\cite{dobzinski07} allows us to overcome this obstacle: for any given allocation, the lemma  makes a polynomial number of demand queries and obtains a sub-allocation that can be broken down into a `pricing-friendly'-format. Following this, we show how to use this lemma to design a OTF mechanism that provides the desired approximation guarantees.

\begin{lemma}
\label{lem_polyt}
Given a subadditive function $v$ and a set of goods $A \subseteq \mathcal{I}$, we can compute in poly-time a price $\hat{p}$ and a sub-allocation $B \subseteq A$ such that $(i)$ for any $T \subseteq \mathcal{I}$, $v(B \setminus T) \geq \hat{p}|B \setminus T|$, and $(ii)$ $2e\log(M)\hat{p}|B| \geq v(A).$
\end{lemma}
\begin{proof}

We begin by recursively define the sets $B^0 \supseteq B^1 \supseteq \ldots \supseteq B^r$ as follows: first let $B^0 = A$, now let

$$p_t = \frac{v(B^{t-1})}{|B^{t-1}|\log(M)};$$ 

Now $B^t$ is (recursively defined to be) the subset of $B^{t-1}_j$ that maximizes $v(T) - p_t|T|$, i.e., $B^{t} := \argmax_{T \subseteq B^{t-1}} (v(T) - p_t|T|)$. Such a set can easily be found using a single demand query for valuation $v$ setting the price of items outside $B^{t-1}$ to be infinite. \\

We stop this recursive procedure when $|B^t| \geq \frac{|B^{t-1}|}{2}$. Therefore, for every $t < r$, $|B^t| < \frac{|B^{t-1}|}{2}$ and $|B^r| \geq \frac{|B^{r-1}|}{2}$. It follows that since $|B^0| = |A| \leq M$ and the number of items is (more than) halved in each iteration, $r \leq \log(M)$. Set $B=B^r$ and $\hat{p}=p_r$. We are now ready to prove the two parts of our lemma using this $B$ and $\hat{p}$. \\

\noindent \textbf{(Part I)} First, we need to show that for all $T\subseteq \mathcal{I}$, $v(B \setminus T) \geq \hat{p}|B \setminus T|$. Assume by contradiction that $\exists T$ not satisfying this condition, i.e., $v(B \setminus T) < \hat{p}|B \setminus T|$. 

Using subadditivity, $v(B \cap T) - \hat{p}|B\cap T| \geq \{v(B) - \hat{p}|B|\} - \{v(B \setminus T) - \hat{p}|B \setminus T|\} > v(B) - \hat{p}|B|$. However, this implies that it is $B \cap T$ and not $B = B^r$ that would be returned by the demand query at price $\hat{p}$. This contradicts the definition of $B^r$ as described in the previous paragraphs.

\noindent \textbf{(Part II)} Next, we prove that $\hat{p}|B| \geq \frac{v(A)}{2e\log(M)}$. We claim that for every $t$, $v(B^t) \geq v(B^{t-1})(1-\frac{1}{\log(M)})$. Indeed, during iteration $t$, the recursive procedure chose $B^t$ as the utility maximizing allocation for every buyer. This means that the utility due to this allocation is at least the utility due to the allocation $B^{t-1}$ itself, and so 

\begin{align*}
v(B^t) & \geq v(B^t) - p_t|B^t| \\
 & \geq v(B^{t-1}) - p_t|B^{t-1}|\\
& = v(B^{t-1})) - \frac{v(B^{t-1})}{\log(M)}
\end{align*}

The last equality comes from the definition of $p_t$. Telescoping this argument gives us $v(B^{r-1}) \geq v(B^0)(1-\frac{1}{\log(M)})^{\log(M)-1} \geq \frac{v(B^0)}{e}$. Since $\hat{p}=\frac{v(B^{r-1})}{|B^{r-1}|\log(M)}$ and $|B^r| \geq \frac{|B^{r-1}|}{2}$, we get the final bound that 
$$\hat{p}|B^r| \geq \frac{1}{2e} \frac{v(A)}{\log(M)}.$$
\end{proof}

The lemma leads to the following straightforward corollary. 

\begin{corollary}
\label{corr_structuralprices}
Given a valuation profile $\vec{v} = (v_1, \ldots, v_N)$ and an allocation of goods $A=(A_1, A_2, \ldots, A_N)$, we can compute in poly-time a price vector $(\hat{p}^j)_{j=1}^N$ with one price per buyer and a sub-allocation $B=(B_1, \ldots, B_N)$ such that $(i)$ for any buyer $j$ and $T \subseteq \mathcal{I}$, $v_j(B_j \setminus T) \geq \hat{p}^j|B_j \setminus T|$, and $(ii)$ $2e\log(M) \sum_{j=1}^N \hat{p}^j|B_j| \geq SW(A).$
\end{corollary}

Using the above structural lemma as our base, we will now define and analyze our OTF mechanism for subadditive buyers. Suppose that $Alg$ is the algorithm that is input to the theorem, and let $Alg(\vec{v})$ be the allocation returned for input profile $\vec{v}$. Additionally, we define an augmented alg $\overline{Alg}$ that works as follows: given as input a valuation profile, the augmented algorithm first runs $Alg$ on the input profile $\vec{v}$ to obtain an allocation $Alg(\vec{v})$. Following this, we run our poly-time procedure from Corollary~\ref{corr_structuralprices} on $Alg(\vec{v})$ to obtain a reduced allocation $B(\vec{v})$, which is the final output of our augmented algorithm. Recall that the augmented algorithm also returns a price vector $\hat{p}^j(\vec{v})$ for $j=1$ to $N$. 

The rest of the proof is quite similar to that of Lemma~\ref{lem_structural} except that we will now use $B(\vec{v})$ as our benchmark allocation. We also use the same notation, albeit analogously for $B(\vec{v})$, i.e., $N_i(\vec{v})$ is the set of users to whom good $i$ is allocated in $B(\vec{v})$, and $k_i(\vec{v}) = |N_i(\vec{v})|$. Using the per-buyer price vector, we can also define the social welfare due to each good. Specifically, for good $i \in \mathcal{I}$, let $SW_i(\vec{v}) := \sum_{j \in N_i(\vec{v})}\hat{p}^j(\vec{v})$, and $SW_i = E_{\vec{v} \sim \mathcal{F}}[SW_i(\vec{v})]$.  Finally, we recall (from Corollary~\ref{corr_structuralprices}) that $SW(Alg(\vec{v})) \leq O(\log(M))\sum_{i \in \mathcal{I}}SW_i(\vec{v}).$ 

\subsubsection*{OTF Mechanism}


We now propose an OTF mechanism that uses a single price $\hat{p}_i$ on each good $i$ and obtains a $O(\log(M))$-approximation to the welfare of $Alg$. Consider the mechanism, whose price on good $i$ is 

$$\hat{p}_i = \frac{1}{2E_{\vec{v}\sim \mathcal{F}}[k_i(\vec{v})]}SW_i.$$ The supply constraints on our mechanism are the natural supply constraints of the problem itself. 

It remains to show that the given mechanism obtains a good approximation to our original welfare. We follow the same template as the proof of Lemma~\ref{lem_structural} making minor changes accordingly. Fix valuations $\vec{v}$ and consider any $\vec{a} = (v_j, \vec{a}_{-j})$ as we did before, and define $Sold_j(\vec{v})$ accordingly. 

Under the fixed valuation $\vec{v}$ and for each $\vec{a}_{-j}$, the buyer $j$ could have purchased the set of goods in $B_j(\vec{a}) \setminus Sold_j(\vec{v})$. Therefore, the utility of this buyer under the fixed valuation is at least $v_j(B_j(\vec{a}) \setminus Sold_j(\vec{v})) - \sum_{i \in B_j(\vec{a}) \setminus Sold_j(\vec{v})}\hat{p}_i$. However, note from our definitions and corollary~\ref{corr_structuralprices} that $v_j(B_j(\vec{a}) \setminus Sold_j(\vec{v})) \geq \hat{p}^j(\vec{a}) | B_j(\vec{a}) \setminus Sold_j(\vec{v})|$. Therefore, taking the expectation over every $\vec{a} \sim \mathcal{F}$, we get the following lower bound on the utility of agent $j$ under the fixed valuation,

$$E_{\vec{a}_{-j}}[\hat{p}^j(\vec{a})|B_j(\vec{a}) \setminus Sold_j(\vec{v}))| - \sum_{i \in B_j(\vec{a}) \setminus Sold_j(\vec{v}))}\hat{p}_i].$$


Therefore, taking the expectation over all $\vec{v}$ and adding the surplus for all buyers, we can bound the total buyer utility as follows

\allowdisplaybreaks
\begin{align*}
E_{\vec{v} \sim \mathcal{F}}[\sum_{j =1}^N u_j(\vec{v})] & \geq \sum_{j=1}^N E_{\vec{a}_{-j}}[\sum_{i \in \mathcal{I}} \mathbbm{1}{(j \in N_i(\vec{a}))}.  \mathbbm{1}(i \notin Sold_j(\vec{v})).\hat{p}^j(\vec{a}) - \hat{p}_i] \\
& = \sum_{i \in \mathcal{I}}\sum_{j =1}^N E_{\vec{v}}[\mathbbm{1}(i \notin Sold_j(\vec{v}))] E_{\vec{a}} [ \mathbbm{1}(j \in N_i(\vec{a})).\hat{p}^j(\vec{a}) - \hat{p}_i ] \\
& \geq \sum_{i \in \mathcal{I}}Pr(\text{i is not sold out}) E_{\vec{v}}[SW_i(\vec{v}) - \hat{p}_ik_i(\vec{v})] \\ 
& = \sum_{i \in \mathcal{I}}Pr(\text{i is not sold out}) SW_i - \hat{p}_iE_{\vec{v}}[k_i(\vec{v})] \\ 
& = \sum_{i \in \mathcal{I}}Pr(\text{i is not sold out}) \frac{SW_i}{2}. \\ 
\end{align*}

The final line comes from the definition of $\hat{p}_i$. Similar to the previous proof, we can also show that the total profit is at least the profit due to the sold items and therefore, bounded from below by 
$$\sum_{i \in \mathcal{I}}Pr(\text{i is sold out})k_i\hat{p}_i \geq \sum_{i \in \mathcal{I}}\hat{p}_iE_{\vec{v}}[k_i(\vec{v})] \geq \sum_{i \in \mathcal{I}}Pr(\text{i is sold out})\frac{SW_i}{2}.$$

The rest of the theorem follows from adding the profit and surplus terms. $\qed$.

Now in order to obtain an actual mechanism, we define $Alg$ to be the $2$-approximation algorithm for the allocation problem with subadditive buyers by Feige~\cite{feige09}. Technically, the algorithm is defined for unit-supply. However, it can easily be extended to the limited supply case (supply of each good is at most $N$) by treating each copy of a good as a separate good. 
\end{proof}

\end{document}